\documentclass[11pt]{article}
\usepackage[margin=1in]{geometry}
\usepackage{amsfonts}
\usepackage{amsmath,amsthm,amssymb}
\usepackage{latexsym}
\usepackage{epic}
\usepackage{epsfig}
\usepackage{hyperref}
\usepackage{verbatim}
\usepackage[justification=centering]{caption}
\usepackage{enumitem}
\usepackage{color}
\allowdisplaybreaks[1]
\usepackage[numbers]{natbib}
\usepackage{setspace}
\usepackage{bookmark}
\usepackage{subfiles}
\usepackage{bbm}
\usepackage[normalem]{ulem}

\usepackage{array}
\usepackage{adjustbox}
\usepackage{multirow}
\usepackage{makecell}
\usepackage{booktabs}

\usepackage{tikz}
\usepackage{pgfplots}
\usepackage{pgfplotstable}
\newcolumntype{P}[1]{>{\centering\arraybackslash}p{#1}}
\newcolumntype{M}[1]{>{\centering\arraybackslash}m{#1}}

\usepackage{algpseudocode}
\usepackage{todonotes}
\usepackage{algorithm}

\usepackage{import}
\usepackage[utf8x]{inputenc}

\newtheorem{theorem}{Theorem}[section]
\newtheorem{corollary}[theorem]{Corollary}
\newtheorem{lemma}[theorem]{Lemma}
\newtheorem{proposition}[theorem]{Proposition}
\newtheorem{claim}[theorem]{Claim}
\newtheorem{definition}[theorem]{Definition}

\newtheorem{question}[theorem]{Question}
\newtheorem{observation}[theorem]{Observation}

\newtheorem{example}[theorem]{Example}

\newcommand{\set}[1]{\left\{ #1 \right\}}
\newcommand{\union}{\cup}
\newcommand{\intersect}{\cap}

\renewcommand{\hat}{\widehat}
\renewcommand{\tilde}{\widetilde}
\renewcommand{\bar}{\overline}

\newcommand{\bvec}[1]{\boldsymbol{ #1 }}

\newcommand{\ce}{\texttt{C} }
\newcommand{\nc}{\texttt{NC}}
\newcommand{\aug}{\texttt{AUG}}
\newcommand{\crs}{\mathtt{CRS}}
\newcommand{\greedy}{\mathtt{Greedy}}
\newcommand{\crsM}{\texttt{CRS-BaseMatching}}
\newcommand{\match}{\langle G,p \rangle}

\DeclareMathOperator{\opt}{\textbf{OPT}}
\DeclareMathOperator{\optsol}{opt}
\DeclareMathOperator{\sopt}{\mathcal{D}_{opt}}
\DeclareMathOperator{\asopt}{\tilde{\mathcal{D}}_{opt}}

\DeclareMathOperator{\poly}{poly}

\DeclareMathOperator{\spn}{\textbf{Span}}

\newcommand{\Exp}{\mathbb E}
\newcommand{\rank}{\textbf{Rank}}

\let\vec\mathbf

\def\min{\qopname\relax n{min}}
\def\max{\qopname\relax n{max}}

\def\argmax{\qopname\relax n{argmax}}

\def\Pr{\qopname\relax n{\mathbf{Pr}}}
\def\Ex{\qopname\relax n{\mathbb{E}}}

\newcommand{\RR}{\mathbb{R}}
\newcommand{\RRp}{\RR_+}

\def\A{\mathcal{A}}

\def\D{\mathcal{D}}
\def\E{\mathbb{E}}

\def\I{\mathcal{I}}

\def\M{\mathcal{M}}

\def\O{\mathcal{O}}
\def\P{\mathcal{P}}
\def\Q{\mathcal{Q}}
\def\R{\mathbb{R}}
\def\S{\mathcal{S}}
\def\T{\mathcal{T}}

\def\eps{\epsilon}
\def\sse{\subseteq}

\renewcommand{\vec}{\mathbf}
\newcommand{\spack}{\langle E,\I, f,p \rangle}
\newcommand{\eat}[1]{}

\newenvironment{lp*}{\begin{equation*}  \begin{array}{lll}}{\end{array}\end{equation*}}

\usepackage{thmtools}
\usepackage{thm-restate}
\usepackage{cleveref}

\usepackage{makecell}

\title{On Sparsification of Stochastic Packing Problems}

 \author{
     Shaddin Dughmi\\ 
     Department of Computer Science\\
     University of Southern California\\
     {\tt shaddin@usc.edu}
 \and
     Yusuf Hakan Kalayci\\ 
     Department of Computer Science\\
     University of Southern California\\
     {\tt kalayci@usc.edu}
 \and
     Neel Patel\\
     Department of Computer Science\\
     University of Southern California\\
     {\tt neelbpat@usc.edu}
 }
\date{}

\begin{document}

    \maketitle

    \begin{abstract}
     Motivated by recent progress on stochastic matching with few queries, we embark on a systematic study of the sparsification of stochastic packing problems more generally. Specifically, we consider packing problems where elements are independently active with a given probability~$p$, and ask whether one can (non-adaptively) compute a ``sparse'' set of elements  guaranteed to contain an approximately optimal solution to the realized (active) subproblem.  We seek structural and algorithmic results of broad applicability to such problems.
Our focus is on  computing sparse sets containing on the order of $d$ feasible solutions to the packing problem, where~$d$ is linear or at most polynomial in $\frac{1}{p}$. Crucially, we require $d$ to be independent of the number of elements, or any parameter related to the ``size'' of the packing problem.  We refer to $d$ as the ``degree'' of the sparsifier, as is consistent with graph theoretic degree in the special case of matching.

First, we exhibit a generic sparsifier of degree $\frac{1}{p}$ based on contention resolution. This sparsifier's approximation ratio matches the best contention resolution scheme (CRS) for any packing problem for additive objectives, and approximately matches the best monotone CRS for submodular objectives. Second, we embark on outperforming this generic sparsifier for additive optimization over matroids and their intersections, as well as weighted matching. These improved sparsifiers feature different algorithmic and analytic approaches, and have degree linear in $\frac{1}{p}$. In the case of a single matroid, our sparsifier tends to the optimal solution. In the case of weighted matching, we combine our contention-resolution-based sparsifier with technical approaches of prior work to improve the state of the art ratio from 0.501 to 0.536. Third, we examine packing problems with submodular objectives. We show that even the simplest such problems do not admit sparsifiers approaching optimality. We then outperform our generic sparsifier for some special cases with submodular objectives.
    \end{abstract}

\newpage
\tableofcontents
\newpage

 \section{Introduction}

 Our starting point for this paper is the beautiful line of recent work on variants of the stochastic matching problem, seeking approximate solutions with limited query access to the (stochastic) data~\cite{blum2015ignorance,assadi2017stochastic, assadi2018towards,assadi2019stochastic, behnezhad2019stochastic,behnezhad2020unweighted,behnezhad2020weighted,behnezhad2022vertexcover}. Notably, many of the algorithms in these works are non-adaptive, and can therefore be interpreted as ``sparsifiers'' for the stochastic problem. These works feature powerful new algorithmic and analytic sparsification techniques of possibly more general interest, suggesting that effective sparsifiers might exist well beyond matching and closely related problems.
 Our goal in this paper is to coalesce a broader agenda on the sparsification of combinatorial stochastic optimization problems more generally, beginning with the natural and broad class of packing problems. We ask, and make progress on, the fundamental questions: For which stochastic packing problems is effective sparsification possible? What are the algorithmic techniques and blueprints which are broadly applicable? What are the barriers to progress?

Concretely, we examine stochastic packing problems (SPPs) of the following (fairly general) form. We are given a set system $(E,\I)$, where $E$ is a finite set of \emph{elements} and $\I \sse 2^E$ is a downwards-closed family of \emph{feasible sets} (often also referred to as \emph{independent sets}, in particular for matroids). Also given is an objective function $f: 2^E \to \RRp$, which we assume to be either additive (a.k.a. modular) or submodular. The stochastic uncertainty is described by a given probability~$p \in [0,1]$:  We assume that each element of $E$ is \emph{active}, i.e., viable for being selected, independently with probability $p$. The goal of the SPP is to select a feasible set of active elements maximizing the objective function.

When the set $R$ of active elements is given, or can be queried without restriction, this reduces to non-stochastic optimization for the induced subproblem on $R$. We refer to the output of such an omniscient [approximation] algorithm as an [approximate] \emph{stochastic optimum solution}. We are instead concerned with algorithms that approximate the stochastic optimum by querying the activity status of only a small, a.k.a. ``sparse'', set of elements $Q \sse E$. In particular, as in much of the prior work we require the queried set $Q$ to be chosen non-adaptively. Such algorithms can equivalently be thought of as factoring into two steps: First, a \emph{sparsification algorithm} (or \emph{sparsifier} for short) computes a (possibly random) set of elements $Q$. Second, we learn $R \intersect Q$, and an [approximate] optimization algorithm is applied to the (now fully-specified) subproblem induced by $R \intersect Q$.  Since the second (optimization) step is familiar and well-studied, our focus is on the first step, namely sparsification.

We evaluate a sparsifier by two quantities. The first quantity is a familiar one, namely its \emph{approximation ratio}. Specifically, a sparsifier is $\alpha$-approximate if it guarantees an $\alpha$-approximation to the stochastic optimum solution when combined with a suitable algorithm in the second (optimization) step. The second quantity is a measure of the ``sparsity'' of the set $Q$ selected by the sparsifier. We say our sparsifier is of \emph{degree} $d$ if it guarantees $\Ex [ |Q| ] \leq d \cdot r$, where $r = \max\{|S| : S \in \I\}$ is the \emph{rank} of the set system $(E,\I)$.  Intuitively, the sparsification degree refers to the level of ``contingency'' or ``redundancy'' in the sparsified instance, relative to the size of maximal feasible solutions. Loosely speaking, the degree of a sparsifier roughly measures ``how many'' feasible solutions are maintained to account for uncertainty in the problem. Somewhat fortuitously, our definition of degree specializes to the (average)  graph-theoretic degree in the special case of matching, lending consistency with prior work on stochastic matching with few queries.

We study sparsifiers whose degree admits an upperbound that is independent of the size of the system;  The degree bound can not depend on the number of elements or the rank of the set system, for example.  We focus especially on the ``polynomial regime'', where the degree is restricted to be at most polynomial in~$\frac{1}{p}$. We pursue sparsifiers which are constant-approximate, or in the best case $(1-\epsilon)$-approximate for arbitrarily small $\epsilon > 0$.

 


\subsection*{Results and Techniques}

We begin with the observation that a degree of at least $\frac{1}{p}$ is necessary for constant-approximate sparsification, even for the simplest of packing problems: a rank one matroid and the unweighted additive objective. We then establish a ``baseline'' of possibility for all stochastic packing problems, through a generic sparsifier with this same degree $\frac{1}{p}$. This sparsifier is simple: it computes (or estimates) the marginals $\set{q_e}_{e \in E}$ of the stochastic optimum solution, and outputs a set $Q$ which includes each element $e$ independently with probability $\frac{q_e}{p}$. For SPPs with an additive objective, we show that this sparsifier's approximation ratio matches the balance ratio of the best contention resolution scheme (CRS)\footnote{This is equal to the set system's \emph{correlation gap}, as shown by \cite{chekuri2014submodular}.} for the set system. When the objective is submodular, we approximately match the balance ratio of the best \emph{monotone} CRS up to a factor of $1- \frac{1}{e}$.  We note that contention resolution is only used as a proof tool to certify our sparsifier's approximation guarantee, and is not invoked algorithmically. In settings where the marginals $\set{q_e}_{e \in E}$ are intractable to compute, this sparsifier can be made computationally efficient by resorting to approximation, in which case its approximation ratio degrades in the expected manner. This generic result yields constant-approximate sparsifiers of degree $\frac{1}{p}$ for a large variety of set systems for which contention resolution has been studied, including matroids and their intersections. 

Next, we  embark on ``beating'' this contention resolution baseline for natural SPPs. We succeed at doing so for additive (weighted) optimization over matroids, matroid intersections, and matchings. For a single matroid, we derive a simple greedy sparsifier which is  $(1-\epsilon)$-approximate and has  degree $\frac{1}{p} \cdot \log(1/\epsilon)$. This sparsifier repeatedly adds a maximum weight independent set of the matroid to the sparse set $Q$, and removes it from the matroid, until the desired degree is reached. Though our sparsifier is simple, its analysis is (we believe necessarily) less so.

For matroid intersections, we first argue that adaptations of our single-matroid sparsifier cannot succeed, due to feasible sets not ``combining well'' as they do in the case of a single matroid. Instead, our sparsifier for matroid intersections repeatedly samples the stochastic optimum solution and adds it to the sparse set $Q$, for a degree of $O(\frac{1}{p\eps} \cdot \log(1/\epsilon))$. The approximation ratio of our sparsifier for the intersection of $k$ matroids is $\frac{1-\eps}{k+ \frac{1}{k+1}}$, which beats the best known bound on the correlation gap of $\frac{1}{k+1}$ \cite{adamczyk2018random}. The analysis of this sparsifier is again nontrivial, and utilizes basis exchange maps.

For both matroids and matroid intersections, we note that analysis techniques employed by prior work on matching do not appear to suffice. In particular, prior work on matching often employs concentration arguments on the active degree of matroid ``flats'' containing an element; this is sufficient in the case of matching, since each element is in at most two binding flats (one for each partition matroid). For general matroids, such concentration arguments fail to bound the degree in a manner independent of the number of elements, necessitating alternative proof approaches like ours.

For general (non-bipartite) matching, we augment our contention-resolution-baseline sparsifier with samples from the stochastic optimum solution, for a total degree of $O( \frac 1 p)$. We show that the samples from the stochastic optimum combine well with our baseline sparsifier. We obtain an approximation ratio which is a function of the (as yet not fully known) correlation gap of the matching polytope. This function exceeds the identity function everywhere, implying that our sparsifier strictly improves on the contention resolution baseline. Plugging in the best known lowerbound of $0.474$ on the correlation gap from \cite{calum-improved-matching-crs}, we guarantee that our sparsifier is $0.536$ approximate. This improves on the state of the art in the polynomial regime, namely the $0.501$-approximate sparsifier of degree $\poly\left( \frac 1  p\right)$ due to \cite{behnezhad2019stochastic}.



Finally, we further examine stochastic packing problems with submodular objectives. Our $(1-\epsilon)$-approximate sparsifier for weighted matroid optimization might tempt one to conjecture a similar result for submodular optimization over simple enough set systems. However, we show by way of an information-theoretic impossibility result that no sparsifier with degree bound independent of the number of elements can beat $(1-1/e)$, even for optimizing a coverage function subject to a uniform matroid constraint. We complement this impossibility result with algorithmic sparsification results for optimizing coverage functions over matroids, improving over the guarantees provided by our baseline generic sparsifier.

\subsection*{Additional Discussion of Related Work}

The exploration of sparsifying SPPs was initiated by \cite{blum2015ignorance}, who focus on the  unweighted stochastic matching problem. They achieve a $\frac {1-\epsilon} 2$-approximate sparsifier with degree $O\left( \frac 1 {p^{1/ \epsilon}}\right)$. This problem has since been studied extensively in a series of works \cite{assadi2019stochastic,behnezhad2019stochastic,behnezhad2020unweighted,behnezhad2020weighted} which attempt to beat the benchmark set by \cite{blum2015ignorance}. In the ``polynomial-degree regime'', the state-of-the-art sparsifier for unweighted stochastic matching is a $0.66$-approximation due to \cite{assadi2018towards}. Recent work by \cite{behnezhad2022vertexcover} improves this approximation to $\frac{e}{e+1}$ for unweighted bipartite matching. For weighted stochastic matching in the polynomial-degree regime, the current best known sparsifier is a $0.501$-approximation due to  \cite{behnezhad2019stochastic}. Going beyond polynomial degree, \cite{behnezhad2020unweighted} constructed a $(1-\epsilon)$-approximate sparsifier with degree $\exp(\exp(\exp(1/p)))$ for  unweighted stochastic matching. Later on, a similar $(1-\epsilon)$-approximate sparsifier is constructed for the weighted problem by \cite{behnezhad2020weighted}.  The sparsifiers designed for the stochastic matching problems rely heavily on structural properties particular to matching. Our techniques, on the other hand, are targeted at more general packing problems.

To the best of our knowledge, the work of \cite{maehara-adaptive,maehara2019submodular} stands alone in directly studying the sparsification of SPPs beyond matching. In \cite{maehara-adaptive}, they proposed a general framework for solving stochastic packing integer programs. As a corollary of their techniques, they obtain non-adaptive sparsifiers for several additive SPPs. However, the degree of their sparsification algorithms intrinsically depends on the number of elements in settings where a single element may be in an exponential number of binding constraints (as is the case for matroids). Our work, in contrast, proposes several algorithmic techniques that yield approximate sparsifiers with degree independent of the number of elements. 

Also related is the work of \cite{Goemans04coveringminimum}, which studies the covering analogue of our question for matroids. They show how to construct a set of size $O(\frac{\rank}{p}\log \frac{\rank}{\epsilon} )$ which is guaranteed to contain a minimum-weight base of the matroid with high probability. This implicitly implies an $O(\frac 1 p\log \frac{\rank}{\epsilon} )$-degree sparsifier for weighted stochastic packing on matroids. Their analysis is tight for the covering setting, and it appears nontrivial to adapt their techniques for the packing setting in order to remove the degree's dependence on the rank.  We compare our results for additive SPPs with prior work in Table~\ref{tab:additive}.


The manuscript \cite{maehara2019submodular} proposes sparsifiers for SPPs with a monotone submodular objectives. However, their sparsification algorithms are intrinsically adaptive in nature. 
To the best of our knowledge, ours is the first work that analyzes SPPs with submodular objectives in the non-adaptive setting. 
We summarize our results for  submodular SPPs in Table~\ref{tab:submodular}.

\begin{table}[h]
		\centering
    \begin{tabular}{ | M{2cm} || M{1.8cm} | M{4cm} || M{2cm} | M{3cm} |} 
        \hline
        \rule{0pt}{20pt}
         & \multicolumn{2}{c||}{\textbf{Previous Results}} & \multicolumn{2}{c|}{\textbf{This Work}}\\
        \hline
        \textbf{Constraint} &\textbf{Approx. Ratio} & \textbf{Sparsification Degree} &  \textbf{Approx. Ratio} & \textbf{Sparsification Degree} 
        \\ [0.5ex] 
        \hline\hline
        Matroid
         & $1-\epsilon$ \cite{Goemans04coveringminimum} & $O\left(\frac 1 p \log \frac{\rank}{\epsilon} \right)$  & $1-\epsilon$ & $\frac{1}{p} \cdot \log \left(\frac{1}{\epsilon}\right)$\\ 
		\hline

        $2$-Matroid Intersection & $\frac{1-\epsilon}{2}$ \cite{maehara-adaptive} & $O \left(\frac W p \log\left(\frac{n}{\epsilon}\right) \frac{1}{\epsilon p} \right)$ & $(1-\epsilon)\cdot\frac{3}{7}$ & $\frac{1}{\eps \cdot p} \cdot \log \left(\frac{1}{\epsilon}\right)$\\
		\hline
        $k$-Matroid Intersection & $\frac {1-\epsilon} {2k}$ \cite{maehara-adaptive} & $O \left(\frac W p\log n \log \left( \frac n \epsilon \right) \frac{k}{\epsilon^3}\right)$ & $\frac{1-\epsilon}{k+\frac{1}{k+1}}$ & $\frac{1}{\eps \cdot p} \cdot \log \left(\frac{1}{\epsilon}\right)$\\
        	\hline

       General Matching & $1-\eps$ \cite{behnezhad2020weighted}  & $O \left(\exp(\exp (\exp ( 1/p))) \right) $ & $0.536$ & $O\left(\frac{1}{p}\right)$\\
        \hline 

        General Matching & $0.501$ \cite{behnezhad2019stochastic}&  $O\left(\frac 1 p\right)$ & $0.536$ & $O\left(\frac{1}{p}\right)$\\
		\hline
		
    \end{tabular}
	\caption{\label{tab:additive} Summary of information theoretic sparsifiers for additive objectives. Here, $n$ is the number of elements and $W$ is the maximum element weight.}
\end{table}

{\normalfont
	\begin{table}[h]
		\centering
    \begin{tabular}{ | M{2.5cm} || M{3.5cm} | M{3cm} | M{4cm} |} 
		\hline
        \textbf{Constraint} &\textbf{Approximation Ratio} & \textbf{Sparsification Degree} & \textbf{Note}\\
		\hline
        \hline
        $r$-Uniform Matroid & $\left(1-\frac{1}{e}\right) \cdot \left(1-\frac{1}{\sqrt{r+3}}\right)$ & $ \frac{1}{p} $ & $\left(1-\frac{1}{e}\right)$ upperbound, Optimal when $r \rightarrow \infty$\\
		\hline
        Matroid & $\left(1-\frac{1}{e}\right)^2$ & $\frac{1}{p} $ & $1-\frac{1}{e}$ upperbound \\
        \hline
        $k$-Matroid Intersection & $\left(1-\frac{1}{e}\right) \cdot \frac{1}{k+1}$ & $ \frac{1}{p} $ &\\ [1ex] 
        \hline
	\end{tabular}
	\caption{\label{tab:submodular} Summary of information theoretic sparsifiers for monotone submodular objectives. All mentioned results are showed in this paper.}
	\end{table}
}

\section{Preliminaries}
\subsection{Matroid Theory}

We use standard definitions from matroid theory; for details see \cite{oxley06,welsh10}. A \emph{matroid} $\M=(E, \I)$ is a structure consisting of elements $E$ and family of \emph{independent sets} $\I \subseteq 2^E$ satisfying the three \emph{matroid axioms}. A \emph{weighted matroid} consists of a matroid $\M = (E,\I)$ together with weights $w \in \R^E$ on the elements. We use the standard matroid-theoretic notions of \emph{\emph{basis}, \emph{dependent set}, \emph{circuit} and \emph{flat}. As usual, we also define the \emph{rank} of a matroid to be the cardinality of its bases.}
	
Given a matroid $\M=(E, \I)$,  we use $\rank^\M$ to denote the \emph{rank function} of $\M$, where $\rank^\M(S)=\max \{|T| \mid T \subseteq S, T \in \I\}$.  
We also use $\spn^\M$ to denote the \emph{span function} of matroid $\M$, where $\spn^\M(S)=\{e \in E: \rank^\M(S \cup \{e\}) = \rank^\M(S)\}$. When clear from context, we sometimes omit the superscript $\M$ from these functions. We slightly abuse notation and use $\rank(\M) = \rank^\M(E)$ to denote the rank of the matroid $\M$. We say a set $S \sse E$ is \emph{full rank} in $\M$ if $\rank^\M(S) = \rank(\M)$; equivalently, if $\spn^\M(S) = E$.

We employ the standard matroid operations of \emph{deletion}, \emph{restriction}, and \emph{contraction}. For matroid $\M=(E,\I)$ and $S \sse E$, we use $\M \setminus S = (E \setminus S, \I \intersect 2^{E \setminus S})$ to denote \emph{deleting} $S$ from $\M$, and $M | S = (S, \I \intersect 2^S)$  to denote \emph{restricting} $\M$ to $S$.   We also use $\M / S$ to denote \emph{contracting} $\M$ by some $S \in \I$, where $\M / S$ is the matroid with elements $E \setminus S$ and independent sets $\set{ T \sse E \setminus S \mid S \union T \in \I}$. We use the fact that the rank function of the contracted matroid is given by \begin{equation} \rank^{\M/S}(T)=\rank^\M(T \cup S) - \rank^\M(S) = \rank^\M(T \cup S) - |S| \label{eq:rankcontract}.\end{equation}


\subsection{Set Functions}
\label{sec:setfunctions}

We consider nonnegative set functions $f : 2^E \rightarrow \mathbb R_+$. Such a function is \emph{additive} (a.k.a. \emph{modular}) if it can be expressed as $f(S) = \sum_{e \in S} w_e$ for some (nonnegative) \emph{weights} $\vec w$.

More generally, we consider nonnegative set functions $f : 2^E \rightarrow \mathbb R_+$ that are \emph{submodular} and \emph{monotone non-decreasing} (or just \emph{monotone} for short). Submodularity entails that  $f(S\cup T)+f(S\cap T) \leq f(S)+f(T)$ for all $S,T \sse E$. Monotonicity entails that $f(S)\subseteq f(T)$ whenever $S\subseteq T$. 

We use the \emph{multilinear extension} $F:[0,1]^n\rightarrow \mathbb R_+$ of a set function $f$, as originally defined by \cite{calinescu2011maximizing}. Specifically, for each $\vec x \in [0,1]^n$ the mutlinear extension is expressed as follows.
\begin{equation*}
    F(\vec x) = \sum_{S\subseteq E} f(S)\prod_{e\in S}x_e \prod_{e'\in E\setminus S}(1-x_{e'})
\end{equation*}
We also reference the \emph{concave closure} $f^+:[0,1]^n\rightarrow \mathbb R_+$ of a set function $f$,  whose output at $\vec x$ is the maximum expected value of $f$ over distributions with marginals $x$.  More formally, 
\begin{equation*}
    f^+(\vec x) = \max\left\{ \sum_{S\subseteq E}\alpha_S f(S)\mid \sum_{S\subseteq E} \alpha_S\leq 1,\alpha_S\geq 0\ \wedge\ \forall e\in E:\sum_{S:e\in S} \alpha_S \leq x_e \right\}
\end{equation*}

\cite{vondrak2007submodularity} proved the following relation between $F(\cdot)$ and $f^+(\cdot)$ when $f$ is monotone submodular. This relationship has since been interpreted as bounding the \emph{correlation gap} of monotone submodular functions by $1-1/e$. 

\begin{theorem}[\cite{calinescu2011maximizing}]\label{thm:correlation_gap}
    For any nonnegative monotone submodular function $f$, \[F(\vec x)\geq \left(1 - \frac{1}{e} \right)\cdot f^+(\vec x) \] for all $\vec x \in [0,1]^n$.
\end{theorem}




\section{Problem Definition}
We consider packing problems of the form $\langle E,\I, f \rangle$ where $E$ is a ground set of \emph{elements} with cardinality $n$, $f:2^E\rightarrow \R_{\geq 0}$ is an objective function, and $\I\subseteq 2^E$ is a downwards-closed family of \emph{independent sets} (a.k.a. \emph{feasible sets}).
We use $r=\argmax\{|I|:I\in \I\}$ to denote the \emph{rank} of the set system $\I$. The aim of the packing problem is to select an independent set $O\in \I$ that maximizes~$f(O)$.

In this paper, we study packing problems in a particular setting with uncertainty parametrized by $p \in [0,1]$. In a \emph{stochastic packing problem (SPP)} $\spack$, nature selects a random set $R\subseteq E$ of \emph{active elements}  such that $\Pr[e\in R] = p$ independently for all $e\in E$. We are then tasked with solving the induced (random) packing problem on the active elements, namely $\langle R, \I |R, f | R\rangle$ where $\I|R$ and $f|R$ denote the restriction of $\I$ and $f$ to subsets of $R$, respectively. 
We refer to an [approximately] optimum solution  to  $\langle R, \I |R, f | R\rangle$  as an [approximate] \emph{stochastic optimum solution}. We use $\opt$ to denote the expected value of a stochastic optimum solution, i.e.,
$$\opt(E,\I,f,p) = \E_{R}\left[ \max_{\substack{T\in \I\\ T\subseteq R}} f(T) \right],$$ where $R \sse E$ is the random set which each element of $E$ independently with probability $p$.

We assume that the set $R$ of active elements  is a-priori unknown, and that we can \emph{query} elements in $E$ to check their membership in $R$. Motivated by settings in which queries are costly, we seek algorithms which query a small (we say ``sparse'') subset of the elements, and moreover choose those queries non-adaptively. Such non-adaptive algorithms can be thought of as factoring into two steps: A \emph{sparsification} step which selects the small set $Q \sse E$ of queries, and an \emph{optimization} step which solves the packing problem $\langle R\intersect Q, \I |R\intersect Q, f | R\intersect Q\rangle$ induced by the queried active elements.  For the optimization step, we assume access to a traditional [approximation] algorithm. Our focus is on algorithms for the sparsification step, which we define formally next. 

A \emph{sparsification algorithm} (or \emph{sparsifier} for short)  $\A$ takes as input an SPP $J = \spack$ from some family of SPPs, and outputs a (possibly random) set of elements $Q \sse E$. The twin goals here are for $Q$ to be ``sparse'' in a quantified formal sense we describe shortly, while guaranteeing  that optimally solving the ``sparsified''  SPP $J | Q = \langle Q, \I |Q, f |Q, p \rangle$ yields an approximate solution to the original SPP $J$.   We say that the sparsification algorithm $\A$ is \emph{$\alpha$-approximate} if it guarantees $\opt(J|Q) \geq \alpha \opt(J)$ --- i.e., an optimal solution to the sparsified SPP is an $\alpha$-approximate solution to the original SPP. We sometimes identify the sparsified SPP $J|Q$ with $Q$ when $J$ is clear from the context.

  To quantify sparsity, we say that $\A$ has \emph{sparsification degree} $d$ if it guarantees that $$ \frac{\Ex[|Q|]}{r} \leq d,$$ where $r$  is the rank of the set system $\I$, and expectation is over the internal random coins of $\A$. Intuitively, the degree of sparsification refers to the level of ``contingency'' or ``redundancy'' in the sparsified instance, relative to the size of maximal feasible solutions. Loosely speaking, the degree of a sparsifier roughly measures ``how many'' feasible solutions it maintains to account for uncertainty in the problem. 
  



In the absence of a bound on degree, an approximation factor of $\alpha=1$ is trivially achievable. We aim to construct approximate sparsifiers of low degree for natural classes of SPPs. We begin by observing that a degree of $\Omega(1/p)$ is necessary for constant approximation, even for the simplest of constraints.
\begin{example}
\label{ex:simple-impossibility}
Consider the SPP with $n$ elements, the unweighted additive objective function $f(S) = |S|$, a rank-one matroid constraint, and activation probability $p=1/n$. There is at least one active element with probability $1-(1-p)^n \geq 1-1/e$, therefore $\opt \geq 1-1/e$. On the other hand, a set of elements $Q$ will contain no active elements with probability $(1-p)^{|Q|} \geq 1-|Q| \cdot p = 1-\frac{|Q|}{n}$. When $|Q| = o(1/p) = o(n)$, there are no active elements in $Q$ with probability $1-o(1)$.  Therefore, any constant-approximate sparsifier must have degree $\Omega(1/p)$.
\end{example}
We also show in Appendix \ref{sec:ssp_additive_impos} that, unsurprisingly, there exist stochastic packing problems which do not admit constant approximate sparsifiers with degree $\poly(1/p)$. Given these simple impossibility results, we ask a natural question:
\begin{question}
Which stochastic packing problems admit constant approximate sparsifiers of degree~$O\left( \frac 1 p\right)$, or more loosely $\poly\left(\frac 1 p\right)$?
\end{question}

In this paper, we focus on designing sparsification algorithms for stochastic packing problems with additive or nonnegative monotone submodular objectives. 


\subsection*{A Note on Input Representation}

Many of our results are information theoretic, and  therefore make no assumptions on how a stochastic packing problem is represented. Most of our algorithmic results, on the other hand, only require solving realized (non-stochastic) instances of the packing problem, possibly approximately. Specifically, for a stochastic packing problem $\spack$ we often assume access to a [$\beta$-approximate] \emph{stochastic optimal oracle}. Such an oracle samples a [$\beta$-approximate] solution to the  (random) packing problem $\langle R, \I |R, f | R\rangle$, where $R$ includes each element of $E$ independently with probability $p$, and $\I |R$ and $f | R$ denote the restriction of $\I$ and $f$ to $R$ respectively. For our algorithmic results on matroids, we additionally assume access to an independence oracle, as is standard.

\section{Sparsification from Contention Resolution}\label{sec:CRS_sparse}
In this section, we show how to generically derive a sparsifier for a stochastic packing problem from bounds on contention resolution for the associated set system. First, we recall the relevant definition of contention resolution.

\begin{definition}[\cite{chekuri2014submodular}]
 Let $(E,\I)$ be a set system, and let  $P_\I = \operatorname{convexhull}\{\mathbbm{1}_I: I\in \I\}$ denote the associated polytope in $\RR^E$. A \emph{Contention Resolution Scheme (CRS)} $\pi$ for $P_\I$ is a (randomized) algorithm which takes as input a point $x \in \P_\I$ and a set of \emph{active} elements $A \sse E$, and outputs a feasible set of active elements  $\pi_x(A)$; i.e., with $\pi_x(A) \sse A$ and $\pi_x(A) \in \I$. We say $\pi$ is \emph{monotone} if $\Pr[i \in \pi_x(A_1)] \geq  \Pr[i \in \pi_x(A_2)]$ whenever $i \in A_1 \subseteq A_2$.

 We judge a CRS by its \emph{balance ratio} on inputs drawn from a product distribution with marginals in $P_\I$  --- we refer to such product distributions as \emph{ex-ante feasible} or \emph{feasible on average} for constraints $\I$. Formally, for  $x \in \P_I$ let $R(x) \sse E$ denote the random set which includes each element $i \in E$ independently with probability $x_i$. For $c \in [0,1]$, we say a CRS $\pi$ is \emph{$c$-balanced} if $\Pr [ i \in \pi_x(R(x)) | i \in R(x) ] \geq c$ for all $x \in P_\I$ and $i \in E$. 

\end{definition}

Our generic sparsifier is randomized, has degree $\frac{1}{p}$, and is shown in Algorithm~\ref{alg:CRS_to_Sparse}. We note that our sparsifier computes estimated marginals $\vec q$ for the stochastic optimum solution. For an information-theoretic result, we can assume these to be exact. When the objective function $f$ is additive, this yields a sparsifier with approximation factor matching the balance ratio of the best contention resolution scheme for $P_\I$.\footnote{This balance ratio is equal to the \emph{correlation gap} of the set system $\I$, as shown in \cite{chekuri2014submodular}.} When $f$ is a nonnegative monotone submodular function, the approximation factor matches the balance ratio of the best \emph{monotone} contention resolution scheme for $P_\I$.

To make our sparsifier algorithmically efficient, $\vec q$ may be estimated by sampling from a (possibly approximate) stochastic optimum oracle, in which case our guarantees degrade in the expected manner due to sampling errors and/or the approximation inherent to the oracle.


\begin{algorithm}
    \caption{Generic Sparsifier for a Stochastic Packing Problem $\spack$}\label{alg:CRS_to_Sparse}
    \textbf{Input: }Stochastic packing problem $\spack$
    
    \begin{algorithmic}
      \State Compute the marginals $\vec q$ of the stochastic optimum solution, or an approximation thereof.
    \State $Q\gets \emptyset$;
    \For{all $e\in E$}
         \State Add $e$ to $Q$ with probability $\frac {q_e} {p}$ (independently)
    \EndFor
    \State \textbf{Output: }Sparse set $Q$.
    \end{algorithmic}
\end{algorithm}

\begin{theorem}\label{thm:sparse_CRS_exact}  
  Consider Algorithm~\ref{alg:CRS_to_Sparse}, implemented with exact (possibly non-polynomial-time) computation of the marginals $\vec q$. When $f$ is additive, and $P_\I$ admits a $c$-balanced CRS, the algorithm is a $c$-approximate sparsifier of degree $\frac{1}{p}$. When $f$ is a nonnegative monotone submodular function, and $P_\I$ admits a $c$-balanced monotone CRS, the algorithm is a $c\left(1-\frac{1}{e}\right)$-approximate sparsifier of degree $\frac{1}{p}$.
\end{theorem}
\begin{proof}
  Since the stochastic optimum solution is a (random) feasible set, it is clear that $\vec q  \in P_\I$, and therefore also that $\sum_e q_e \leq r$ where $r$ is the rank of set system $\I$.

  Let $Q$ be the output of Algorithm~\ref{alg:CRS_to_Sparse}. We have $\Ex [ |Q| ] = \sum_{e \in E} \frac{q_e}{p} \leq \frac{1}{p} \cdot r$. Therefore our sparsifier has degree $\frac{1}{p}$. 
    
We now analyze the approximation guarantee. Let $R \sse E$ be the (random) set of active elements, containing each element of $E$ independently with probability $p$. Since $R$ and $Q$ are independent, we have  $$\Pr[e\in Q\cap R] = \Pr[e\in Q]\cdot \Pr[e\in R] = \frac{q_e}{p} \cdot p = q_e.$$ Moreover, since $Q$ and $R$ both follow a product distribution, it follows that $Q \intersect R$ follows the product distribution with marginals $\vec q \in \P_\I$.
   
The existence of $c$-balanced CRS $\pi$ for $P_\I$ guarantees a subset $\pi_q(Q\cap R)\subseteq Q\cap R$ with $\pi_q(Q\cap R)\in \I$ and $\Pr[e\in \pi_q(Q\cap R)] \geq c \cdot \Pr[ e \in Q \cap R] = c \cdot q_e$ for all $e\in E$. When the objective function $f$ is additive with weights $\vec w$, we can lowerbound the optimum value of the sparsified SPP as follows:

  \begin{align*}
    \E\left[\max_{\substack{T\subseteq Q\cap R\\ T\in \I}} f(T)\right]
      &\geq\E\left[f(\pi_q(Q \intersect R))\right] \\
      &= \sum_{e \in E} w_e \cdot \Pr[ e \in \pi_q(Q \intersect R)] \\
      & \geq c\cdot \sum_{e\in E} w_e\cdot q_e \\
      &=c \cdot \opt.
  \end{align*}
 
   Similarly, when the objective function $f$ is monotone submodular, and $\pi$ is a \emph{monotone} $c$-balanced CRS, 
   we can lowerbound the optimum value of the sparsified SPP as follows: 
    \begin{align*}
      \E\left[\max_{\substack{T\subseteq Q\cap R\\T\in \I}} f(T)\right] 
       &\geq\E\left[f(\pi_q(Q \intersect R))\right] \\
       &\geq c\cdot F(\vec q)\\
       &\geq c\cdot \left(1 - \frac{1}{e}\right) \cdot f^+(\vec q)\\
       &\geq c\cdot \left(1 - \frac{1}{e}\right) \cdot \opt
   \end{align*}
   Here $F$ is the multilinear extension of $f$, and $f^+$ is the concave closure of $f$ (see Section~\ref{sec:setfunctions}). The second inequality follows from the fact that $\pi$ is monotone and $c$-balanced, as in \cite[Theorem 1.3]{chekuri2014submodular}.  The third inequality follows from the correlation gap of monotone submodular functions (Theorem~\ref{thm:correlation_gap}). 
   This concludes the proof. 
  \end{proof}

\begin{restatable}{theorem}{crsestimation}\label{thm:sparse_CRS}  
Let $\epsilon > 0$ be a parameter, and assume sample access to a $\beta$-approximate stochastic optimum oracle for $\spack$. An (efficient) implementation of   Algorithm~\ref{alg:CRS_to_Sparse} which uses $\poly(n,\frac{1}{\epsilon})$ samples to estimate $\vec q$ yields the following guarantees.  When $f$ is additive, and $P_\I$ admits a $c$-balanced CRS, the algorithm is a $c\beta\left(1-\epsilon\right)$-approximate sparsifier of degree $\frac{1}{p}$. When $f$ is submodular, and $\P_\I$ admits a $c$-balanced monotone CRS, the algorithm is a $c\beta\left(1-\frac{1}{e}\right) \left(1-\epsilon\right)$-approximate sparsifier of degree $\frac{1}{p}$.
\end{restatable}
The proof of Theorem 4.3 extends that of Theorem 4.2 using standard approximation and sampling bounds, and is therefore relegated to Appendix~\ref{sec:missing_proofs}.
The above theorems, together with contention resolution schemes from prior work \cite{adamczyk2018random,chekuri2014submodular, calum-improved-matching-crs} and approximate stochastic optimal oracles that employ approximation algorithms from \cite{calinescu2011maximizing, lee2010submodular-k-intersect}, imply the following corollary. 

\begin{corollary}\label{coro:crsspars}
    \begin{enumerate}
        \item For an additive objective and a single matroid constraint, there is a $\left ( 1 - \frac 1 e\right )$-approximate information theoretic and $\left(1-\frac 1 e \right)\cdot (1-\epsilon)$-approximate polynomial-time sparsifier with sparsification degree $\frac 1 p $. 
        
        \item For an additive objective and an intersection of two matroid constraints, there is a $\left ( 1 - \frac 1 e\right )^2$-approximate information theoretic and $\left(1-\frac 1 e \right)^2\cdot (1-\epsilon)$-approximate polynomial-time sparsifier with sparsification degree $\frac 1 p $.  
        
        \item  For an additive objective and the intersection of $k$ matroid constraints, there is a $\frac 1 {k+1}$-approximate information theoretic and $\frac{1 -\epsilon}{k^2-1}$-approximate polynomial-time sparsifier with sparsification degree $\frac 1 p $.
        
      \item   For an additive objective and  matching constraints, there is a $0.474$-approximate information theoretic and $(1 -\epsilon) 0.474$-approximate polynomial-time sparsifier with sparsification degree $\frac 1 p $.
        
        \item For a monotone submodular objective and a single matroid constraint, there is a $\left ( 1 - \frac 1 e\right )^2$-approximate information theoretic and $\left(1-\frac 1 e \right)^3\cdot (1-\epsilon)$-approximate polynomial-time sparsifier with sparsification degree $\frac 1 p $. 
        
        \item For a monotone submodular objective and an intersection of two matroid constraints, there is an $\left ( 1 - \frac 1 e\right )^3$-approximate information theoretic and $\left(1-\frac 1 e \right)^4\cdot (1-\epsilon)$-approximate polynomial-time sparsifier with sparsification degree $\frac 1 p $.  
        \item For monotone submodular and the intersection of $k$ matroid constraints, there exists a $\left(1-\frac 1 e \right)\cdot \frac 1 {k+1}$-approximate information theoretic and $\left( 1 - \frac 1 e\right)\frac{1 - \epsilon}{k\cdot (k+1)}$-approximate polynomial-time sparsifier with sparsification degree $\frac 1 p $. 
    \end{enumerate}

\end{corollary}

The following proposition shows that Algorithm~\ref{alg:CRS_to_Sparse} is optimal for matroids and additive objectives among sparsifiers of degree $\frac{1}{p}$. This strongly suggests that sparsification is intimately tied to contention resolution when the degree is restricted to $\frac{1}{p}$. In particular, exceeding degree $\frac{1}{p}$ appears necessary for outperforming the correlation gap of a set system in general.  

\begin{proposition}\label{prop:both-ways}
Consider the family of stochastic packing problems with matroid constraints and additive objectives. There is no degree $\frac{1}{p}$ sparsifier for this family that achieves an approximation ratio  $1 - \frac{1}{e}  + \Omega(1)$.
\end{proposition}
\begin{proof}
Consider the rank one matroid on $n$ elements, with additive objective function $f(S) = |S|$, and let $p=1/\sqrt n$ be the activation probability.  At least one element is active with probability $1-o(1)$, and therefore  $\opt\geq 1-o(1)$. On the other hand, let $Q$ be the (possibly random) set selected by a Sparsification algorithm $\A$ with (expected) degree at most $\frac 1 p$. The probability of observing at least one active element in $Q$ is exactly equal to $1-\left(1 -\frac 1 {\sqrt n} \right)^{|Q|}$. Therefore, we can bound the performance of algorithm $\A$ as follows, where $R$ denotes the set of active elements: 
\begin{align*}
   \E_{Q,R}\left[\max_{T\subseteq Q\cap R. |T|\leq 1} f(T) \right]& = 1-\E\left[\left(1 -\frac 1 {\sqrt n} \right)^{|Q|}\right]\\
   &\leq  1 - \left(1 -\frac 1 {\sqrt n} \right)^{\E[|Q|]}\\
   & \leq 1 - \left(1 -\frac 1 {\sqrt n} \right)^{\sqrt n}\\
   & =  1 -\frac 1 {e} + o(1)
\end{align*}
Where the first inequality is that of Jensen for convex functions. The second inequality follows from the sparsification degree of $\A$.
\end{proof}

We note that $1 - \frac 1 e$ is the best possible balance ratio for contention resolution on the rank one matroid, as shown in \cite{chekuri2014submodular} through the correlation gap. Given the above discussion, it is natural to ask whether we can design sparsifiers of degree $O(1/p)$, or even $\poly(1/p)$, whose approximation ratio  $\alpha$ exceeds the best CRS balance ratio $c$, i.e.,  can we have $\alpha > c$ with degree linear or polynomial in $1/p$? Recent progress on this question for bipartite matching constraints came in a pair of recent works. Behnezhad et.al. \cite{behnezhad2022vertexcover} designed a $\frac{e}{e+1} \approx 0.731$-approximate sparsifier with degree $\poly(1/p)$ for unweighted bipartite matching. Their approximation factor is strictly better than a known upper bound of $0.544$ on the correlation gap (and hence the best balance ratio) of bipartite matching, due to \cite{karp1981maximum}. To our knowledge, this is the only sparsifier in the literature with degree polynomial in $\frac{1}{p}$ and approximation ratio provably exceeding the correlation gap of the set system. Another recent result due to Behnezhad et al \cite{behnezhad2020unweighted} achieves a $0.501$-approximate sparsification with degree polynomial in $1/p$ for \emph{weighted} matching. This outperforms the best \emph{known} contention resolution scheme for matching\cite{bruggmann2020optimal}, though not clearly the best possible. Prior to our work, there was no known sparsifier for any \emph{weighted} stochastic packing problem which provably outperforms the correlation gap using degree $\poly(1/p)$. 


In the following sections, we will construct degree $O(1/p)$ sparsifiers for matroids, matroid intersections, and matching which improve on the contention-resolution-based guarantees provided in this section. For matroids and matchings, our sparsifiers provably outperform contention resolution. For matroid intersections, we outperform the best \emph{known} CRS.


\section{Additive Optimization over a Matroid}
\label{sec:sparse_matroid}

In this section, we design an improved sparsifier for the stochastic packing problem $\spack$ when $\M=(E,\I)$ is a matroid and $f$ is additive. For an arbitrary $\epsilon > 0$, our sparsifier is $(1-\epsilon)$-approximate and has degree $\frac 1 p \log \frac 1 \epsilon$. Throughout, we use  $\set{w_e}_{e \in E}$ to denote the weights associated with the additive function $f$, and use $R \sse E$ to denote the (random) set of active elements which includes each $e \in E$ independently with probability $p$. We also sometimes use $r$ as shorthand for $\rank(\M)$.

\begin{algorithm}[H]
    \caption{Sparsifier for $(\M, f, p)$, when $\M$ is a matroid and $f$ is additive \label{alg:matroid_sparsify}}
    \label{alg:weighted-matroid}
    \begin{algorithmic}
    \State Set $\tau=\frac{1}{p}\cdot \log(\frac{1}{\epsilon})$.
    \State Set $\M_0=\M$.
    \State Set $Q= \emptyset$
    \For{$t$ in $\{1, \dots, \tau\}$}
        \State Let $I_t \leftarrow \argmax_{I \in \I_{t-1}} f(I)$, where $\I_{t-1}$ is the collection of independent sets in $\M_{t-1}$.
        \State Update $\M_t \leftarrow \M_{t-1} \setminus I_{t}$.
        \State Update $Q \leftarrow Q \union I_t$.
    \EndFor
    \State \textbf{Output:} $Q$
    \end{algorithmic}
\end{algorithm}

The following theorem is the main result of this section. 

\begin{theorem}
    \label{thm:matroid_sparsifier}
    Let $\M = (E,\I)$ be a matroid, $f$ be an additive function and $p\in[0,1]$.
    Algorithm~\ref{alg:matroid_sparsify} is a $(1-\epsilon)$-approximate polynomial time sparsifier for the stochastic packing problem $\langle E, \I, f, p\rangle$ with sparsification degree $\frac{1}{p} \cdot \log \left(\frac{1}{\epsilon}\right)$.
\end{theorem}

Previously, the best known non-adaptive result was the $(1-\epsilon)$-approximate sparsifier with degree $O\left(\frac{1}{ p}\log \left(\frac{\rank}{\epsilon}\right) \right)$ implicit in \cite{Goemans04coveringminimum}. In contrast, the sparsification degree of our algorithm is independent of the ``size'' of the matroid. As we argued in our introduction, such a size-independent guarantee appears to be beyond the techniques used in earlier works \cite{Goemans04coveringminimum,maehara-adaptive}.  

It is clear that the sparsifier in Algorithm~\ref{alg:weighted-matroid} has degree $\tau=\frac{1}{p}\cdot \log(\frac{1}{\epsilon})$, and can be implemented in polynomial time given an independence oracle for the matroid $\M$.  The remainder of this section is devoted to proving that it is $(1-\epsilon)$-approximate, as needed to complete the proof of Theorem~\ref{thm:matroid_sparsifier}. Our proof will consist of two parts. First, we will analyze Algorithm~\ref{alg:weighted-matroid} in the special case of unit weights (a.k.a. unweighted). Second, we reduce the analysis of the weighted problem to that of the unweighted problem.

\subsection{Special Case: Unweighted Optimization}


In this subsection, we assume that elements of the matroid $\M$ all have unit weight. In this case, observe that Algorithm~\ref{alg:matroid_sparsify} repeatedly removes an arbitrary basis of the matroid and adds it to the sparse set $Q$. More precisely, in iteration $t$ the set $I_t$ is a basis of the remaining matroid $\M_{t-1}:= \M \setminus \bigcup_{j=1}^{t-1}I_j$.

  In this unweighted case, the stochastic optimal value is the expected rank of the active elements $R$, and our claimed approximation guarantee can be expressed as $\Ex[ \rank(Q \intersect R) ] \geq (1-\epsilon) \Ex[ \rank(R)]$. To establish this, consider the following informal (but ultimately flawed) argument, starting with the observation that $I_t \cap R$ spans a $p$ fraction of the rank of the remaining matroid  $\M_{t-1}$ in expectation. This observation suggests that the rank of elements not spanned by $Q \intersect R$  should shrink by a factor of $(1-p)$ with each iteration. Induction would then guarantee that after $\frac{1}{p}\cdot \log \left(\frac{1}{\epsilon}\right)$ iterations we have covered a $(1-\epsilon)$ fraction of the rank of the matroid.

The above rough argument is a good starting point. Indeed, it succeeds when all (or many) of the bases $I_1,\ldots,I_\tau$ are full-rank or close to it. These are precisely the scenarios in which $\Ex [\rank(R) ] \approx \rank(\M)$. However, in general $\opt= \Ex[\rank(R)]$ can be significantly smaller than  $\rank(\M)$ --- in the worst case up to a factor of $p$ smaller --- in which case the the rank of $I_t$ may drop precipitously with $t$ and the above inductive analysis falls apart. Such scenarios are not simply outliers that we can assume away: they are unavoidable products of the weighted-to-unweighted reduction we present in the next subsection, and can account for a large fraction of the weighted stochastic optimal. This seems to necessitate a more nuanced proof approach in which we compare $\Ex[\rank(Q \intersect R)]$ with $\Ex[\rank(R)]$. We present such a proof next, built upon the following definitions and structural properties.

\begin{definition}
    A \emph{nested system of spanning sets (NSS)} for a matroid $\M$ is a sequence $I_1, I_2, \dots I_\tau$ of sets such that for any $j \in [\tau]$, $I_j$ is a full rank set of elements in $\M \setminus I_{1:j-1}$, where $I_{1:j-1}=\bigcup_{\ell=1}^{j-1} I_{\ell}$.
\end{definition}

\begin{observation}
    \label{obs:NSS}
    The sets $I_1,\ldots,I_{\tau}$ from Algorithm~\ref{alg:weighted-matroid} are an NSS of $\M$.
\end{observation}


The following lemma states that the property of being an NSS is preserved under contraction. 
\begin{lemma}
    \label{lem:contraction}
Let $\M=(E,\I)$ be a matroid and let $I_1, \dots I_\tau$ be an NSS of $\M$. For an arbitrary independent set $S$ of $\M$,  let $I_j' = I_j \setminus S$ for all $j$. Then, the sequence  $I_1', \dots I_\tau'$ is an NSS of $\M/S$.
  \end{lemma}
    \begin{proof}
         Fix an arbitrary $j \in \set{1, \ldots, \tau}$. It is clear that $I'_j$ is a subset of the elements of $\M / S \setminus I'_{1:j-1}$. It remains to show that $I'_j$ is full rank in $\M / S \setminus I'_{1:j-1}$, as follows.
       \begin{align*}
         \rank^{\M / S}(I'_j) &= \rank^\M (I'_j \union S) - |S| & \mbox{(By \eqref{eq:rankcontract})}\\
                              &= \rank^\M(I_j \union S) - |S| &\mbox{(By definition of $I'_j$)}\\
                              &= \rank^\M((E \setminus I_{1:j-1}) \union S) - |S| &\mbox{($I_j$ is full rank in $\M \setminus I_{1:j-1}$)} \\
                              &= \rank^\M((E \setminus I'_{1:j-1}) \union S) - |S| &\mbox{(By definition of $I'_j$)}\\
                              &= \rank^\M((E \setminus S \setminus I'_{1:j-1}) \union S) - |S| \\
                              &= \rank^{\M/S}(E \setminus S \setminus I'_{1:j-1}) & \mbox{(By \eqref{eq:rankcontract})}\\
                              &= \rank(\M/S \setminus I'_{1:j-1})  &\mbox{(By definition of deletion and $\M/S$)}
  \end{align*}
  \end{proof}

\begin{observation}
    \label{obs:deletion}
   If $I_1,\ldots,I_\tau$ is an NSS of $\M$, then $I_2, \dots I_\tau$ is an NSS of $\M \setminus I_1$.
\end{observation}

Now, we will prove the desired result for unweighted matroids.

\begin{lemma}
    \label{lem:unweighted}
    Let $\M$ be a matroid, and let $I_1, \dots I_\tau$ be an NSS of $\M$. Then,
    \[ \Ex[ \rank(I_{1:\tau} \cap R) ] \geq (1-(1-p)^\tau) \cdot \Ex[\rank(R)] \]
\end{lemma}
\begin{proof}
  Let $E$ denote the elements of $\M$. We will apply induction on $\tau$ to prove this result. The base case of $\tau=0$ is trivial. 

  Consider $\tau \geq 1$. Let $S$ be an arbitrary maximal independent subset of $R \intersect I_1$, and let $\rank'$ denote the rank function of the (random) matroid $\M'=\M / S \setminus I_1$ with elements $E \setminus I_1$. Using~\eqref{eq:rankcontract} we can write

    \begin{equation}
        \rank(R \cap I_{1:\tau}) = \rank(R \intersect I_1 ) + \rank'(R \cap I_{2:\tau}) \label{eq:decomp-rank}  
    \end{equation}
    The expected value of the first term is as follows.
    \begin{equation}
      \label{eq:bound1}
      \Ex[ \rank( R \cap I_1 )] = r \cdot p.      
    \end{equation}

    To bound the second term in expectation, we first condition on $R \intersect I_1$, which also fixes $S$ and $\M'$. It follows from Lemma~\ref{lem:contraction} and Observation~\ref{obs:deletion}, as well as the fact that $S \sse I_1$ is disjoint from $I_{2:\tau}$,  that $I_{2}, \dots I_\tau$ is an NSS of $\M'$. This allows us to invoke the inductive hypothesis to obtain
    \begin{equation*}
      \Ex[ \rank'( R \cap I_{2:\tau} ) ] \geq  (1-(1-p)^{\tau-1}) \cdot \Ex[\rank'(R \setminus I_1)].
    \end{equation*}
    From Equation \eqref{eq:rankcontract} and the definition of $S$ we have that $\rank'(R \setminus I_1) = \rank((R \setminus I_1) \union S) - \rank(S) = \rank(R) - \rank(R \intersect I_1)$. Also using \eqref{eq:bound1}, we obtain
        \begin{align}
          \Ex[ \rank'( R \cap I_{2:\tau} )  ] &\geq  (1-(1-p)^{\tau-1}) \cdot \Ex[\rank'(R \setminus I_1)] \nonumber \\
                                              &= (1-(1-p)^{\tau-1}) (\Ex[ \rank(R)] - \Ex[ \rank(R \intersect I_1)]) \nonumber \\ 
                                              &= (1-(1-p)^{\tau-1}) \Ex[ \rank(R)] - (1-(1-p)^{\tau-1}) \cdot r \cdot p \label{eq:bound2}
    \end{align}

Finally, we combine \eqref{eq:decomp-rank}, \eqref{eq:bound1}, and \eqref{eq:bound2} to conclude
\begin{align*}
  \Ex[\rank(R \cap I_{1:\tau})] &\geq (1-(1-p)^{\tau-1}) \Ex[ \rank(R)] + (1-p)^{\tau-1} \cdot r \cdot p  \\
                                &\geq (1-(1-p)^{\tau-1} + p (1-p)^{\tau-1}) \Ex[ \rank(R)] \\
                                &= (1-(1-p)^\tau) \Ex [\rank(R)]
\end{align*}

\end{proof}

By observation~\ref{obs:NSS}, we get the following corollary of Lemma~\ref{lem:unweighted}.
\begin{corollary}\label{coro:unweighted_matroid}
    Consider a stochastic matroid optimization problem $\langle E,\I, f, p \rangle$ for $p \in [0,1]$  and $f(S)=|S|$ for all $S \sse E$. Algorithm~\ref{alg:matroid_sparsify} is a $(1-\epsilon)$-approximate sparsifier with degree $\frac{1}{p}\cdot \log \left( \frac{1}{\epsilon} \right)$.
\end{corollary}

\subsection{Proof of Theorem~\ref{thm:matroid_sparsifier}}
In this section, we will complete the proof of Theorem~\ref{thm:matroid_sparsifier} by reducing the analysis for a general (weighted) additive function to that of the unweighted case. We order the elements $e_1, \ldots e_n$ in decreasing order of their weights $w_1 \geq
  \ldots \geq w_n$. Without loss of generality we assume $w_n > 0$, and for notational convenience we define $w_{n+1} = 0$.  The following lemma says that if a sparsifier is $\alpha$-approximate for the unweighted problem on elements above any given weight threshold, then it is also  $\alpha$-approximate for the weighted problem.



\begin{lemma}
    \label{lem:reduction}
If a set $Q \subseteq E$ satisfies
    \begin{equation}
        \label{eq:greedy_property}
        \Ex[\rank(Q \cap R \cap \{e_1, \dots e_j\})] \geq (1-\epsilon) \Ex[\rank(R \cap \{e_1, \dots e_j\})] ,
    \end{equation} 
 for all $j \in [n]$ with $w_j > w_{j+1}$, then 
    $$ \Ex[f(\optsol(Q \cap R))] \geq (1-\epsilon) \Ex[f(\optsol(R))]. $$
    Here, we denote $\optsol(S) \in \argmax_{\substack{I \subseteq S\\ I\in \I}} f(I)$, with ties broken arbitrarily.
\end{lemma}

Before we prove Lemma~\ref{lem:reduction}, let us show that the output set $Q$ of Algorithm~\ref{alg:matroid_sparsify} satisfies condition~\eqref{eq:greedy_property}.

\begin{lemma}\label{lem:rank_prefix}
   The output set $Q$ of Algorithm~\ref{alg:matroid_sparsify} satisfies
    $$ \Ex[\rank(Q \cap R \cap \{e_1, \ldots, e_j\})] \geq (1-(1-p)^\tau) \Ex[\rank(R \cap \{e_1, \ldots, e_j\})]$$
    for all $j \in [n]$ with $w_j > w_{j+1}$.
\end{lemma}
\begin{proof}
Let $I_1, \ldots, I_\tau$ be as in Algorithm~\ref{alg:matroid_sparsify}, and recall that $I_t$ is a maximum weight base of $\M_{t-1}=\M \setminus I_{1:t-1}$. Note that $Q= I_{1:\tau}$. 
  
   Fix an index $j \in [n]$ with $w_j > w_{j+1}$, and let $\bar{\M}=(\bar{E}, \bar{\I})$ be the matroid $\M$ restricted to the elements $\bar{E}=\{e_1, \dots e_j\}$. Define $\bar{I}_t = I_t \cap \{e_1 \dots e_j\}$ for all $t \in [\tau]$. By Lemma \ref{lem:unweighted}, it suffices to show that $\bar I_{1}, \dots \bar I_{\tau}$ is a nested system of spanning sets for $\bar{\M}$. 

   To show that, we need to show that $\bar{I}_t$ is full rank in $\bar{\M} \setminus \bar{I}_{1:t-1}$ for all $t \in [\tau]$. Assume for a contradiction that for some $t \in [\tau]$, there is an element $e \in \bar{E} \setminus \bar{I}_{1:t-1}$ such that $e \notin \spn^{\bar{\M}}(\bar{I}_t)$. It follows that $e \notin I_{1:{t}}$. Moreover, we know that $I_t$ is a base of $\M \setminus I_{1:t-1}$ and therefore $e \in \spn^{\M}(I_t)$. There is therefore  a unique cycle $C$ in $I_t \cup \{e\}$ with $e \in C$. As $e \notin \spn^{\M}(\bar{I}_t)$, cycle $C$ contains an element $f \in I_t \setminus \bar{I}_t$. Therefore the weight of $f$ is strictly smaller than the weight of $e$, which contradicts the fact that $I_t$ is a maximum weight base of $\M_{t-1}$.
    
\end{proof}

\begin{proof}[Proof of Lemma~\ref{lem:reduction}]
We bound the expected optimum of the sparsified problem $Q$ as follows.

  \begin{align*}
        \Ex[f(\optsol(Q \cap R))] &= \sum_{i=1}^n w_{i} \cdot \Pr[e_i \in \optsol(Q \cap R)]\\
                              &= \sum_{i=1}^n (w_{i} - w_{i+1}) \cdot \Ex[|\optsol(Q \cap R)  \cap \{e_1, \dots e_i\}|]\\
                              &= \sum_{i: w_i > w_{i+1}} (w_{i} - w_{i+1}) \cdot \Ex[|\optsol(Q \cap R)  \cap \{e_1, \dots e_i\}|]\\
                              &= \sum_{i: w_i > w_{i+1}} (w_{i} - w_{i+1}) \cdot \Ex[\rank(Q \cap R \cap \{e_1, \dots e_i\})|]\\
                              &\geq \sum_{i: w_i > w_{i+1}} (w_{i} - w_{i+1}) \cdot (1-\epsilon) \cdot \Ex[\rank(R \cap \{e_1, \dots e_i\})]\\
          &= (1-\epsilon) \cdot \Ex[f(\optsol(R))] 
    \end{align*}

The second equality is from reversing the order of summation. The fourth equality follows from optimality of the greedy algorithm. The inequality follows from our assumption \eqref{eq:greedy_property}. The last equality follows from optimality of the greedy algorithm and reversing the order of summation again.
\end{proof}

Combining Lemmas~\ref{lem:reduction} and~\ref{lem:rank_prefix}   completes the proof of Theorem~\ref{thm:matroid_sparsifier}.


\section{Improved Sparsifier for Stochastic Weighted Matching }
\label{sec:improved_matching}

In this section, we consider the stochastic packing problem with a weighted additive objective function and a general matching constraint, widely known as the stochastic weighted matching problem \cite{blum2015ignorance,behnezhad2019stochastic,behnezhad2020stochastic}. In the instance of stochastic weighted matching $\spack$, the elements $E$ are the edges of a known weighted graph $G:=(V,E,w)$, $\mathcal I$ is the set of all matchings in the graph $G$, and $f$ is an additive function with element weights $\{w_e\}_{e\in E}$. 

For simplicity, we sometimes denote the stochastic matching instance $\spack$ by $\match$ when it is clear from the context. Notice that the optimum value of a stochastic matching instance $I=\match$ is the expected value of the maximum weighted matching on the set of active edges $R$ of graph $G$. The aim of a sparsifier for this problem is to query a $\poly \left(\frac 1 p \right)$-degree subgraph $H$ of $G$ such that the expected weight of the maximum
matching on active edges of $H$ approximates the optimum value of $\match$. The current state-of-the-art $\poly\left( \frac 1 p\right)$-degree sparsifier for the stochastic weighted matching problem achieves a $0.501$ approximation ratio due to \cite{behnezhad2019stochastic}\footnote{Recent work by \cite{behnezhad2020weighted} constructs $(1-\epsilon)$-approximate sparsifier with degree $\exp (\exp (\exp (1/\epsilon,1/p)))$, however, in this work, we focus on sparsifiers with degree $\poly(1/p)$}.
In this section, we design a new $\poly(1/p)$-degree sparsifier for the stochastic weighted matching that improves the approximation ratio to $0.536$.

\cite{behnezhad2019stochastic} proposed a natural greedy sparsifier which selects $\poly(1/p)$ many independent samples from stochastic optimum oracle. For their proof, they provide a procedure to construct a fractional matching on the set of active queried edges with an expected weight of at least  $0.501\cdot \opt$. Their analysis first divides the edges into two disjoint sets: \emph{crucial} edges and \emph{non-crucial} edges. The crucial edges are likely to be part of the stochastic optimum solution whereas the non-crucial edges are less likely to be in the stochastic optimum solution. As a first step, they construct a (fractional) matching $M_1$ on queried active crucial edges, and another matching $M_2$ on queried active non-crucial edges. Selecting the better of $M_1$ and $M_2$ leads to a ratio of $0.5$. The crux of their analysis lies in the involved augmentation procedure of fractional matchings $M_1$ and $M_2$ to improve the approximation ratio from $0.5$ to $0.501$.

Our sparsifier for stochastic weighted matching is divided into two phases: in the first phase, it samples a set of edges $Q_{\crs}$ via generic sparsifier described in Algorithm~\ref{alg:CRS_to_Sparse}, i.e. each edge $e\in E$ in $Q_{\crs}$ with probability $\frac{q_e}{p}$ independently. In the second phase, independent of the set $Q_\crs$, it selects $T$ many independent samples $Q_1,\dots , Q_T$ from stochastic optimum oracle $\mathcal D_{\optsol}$ which is similar to the algorithm of \cite{behnezhad2019stochastic} \footnote{Note that our algorithm does not query the set of edges sampled in the first phase before sampling the second phase. Hence, our sparsifier is non-adaptive.}. Our Sparsifier is already $0.501$-approximate as the second phase is exactly similar to the sparsification algorithm of \cite{behnezhad2019stochastic}. However, we utilize the edges sampled in the first phase to improve the approximation ratio to $0.536$. Our sparsifier is described in Algorithm~\ref{alg:matching_sparsify}. 

\begin{algorithm}
	\caption{Sparsifier for Weighted Stochastic Matching Problem $\match$ \label{alg:matching_sparsify}}
	\label{alg:weighted-matching}
	\begin{algorithmic}[1]
		\State Compute the marginals $\vec q$ of the stochastic optimum solution, or an approximation thereof.
		\State $Q_\crs\gets \emptyset$ and $Q_\greedy \gets \emptyset$
		\State $T\gets $ $\frac{2000\cdot \log \frac 1 \eps \cdot \log (1/\eps)\cdot }{\eps^4 \cdot p}$.
		\For{all $e\in V$}
		\State Add $e$ to $Q_\crs$ with probability $\frac {q_e} {p}$ (independently)
		\EndFor
		\For{$i$ in $1, \dots, T$}
		\State Sample $Q_i \sim \mathcal D_{\optsol}$ independently.
		\State $Q_\greedy \gets Q_\greedy \cup Q_i$.
		\EndFor
		\State \textbf{Output:} $Q = Q_\crs \cup Q_\greedy$.
	\end{algorithmic}
\end{algorithm}

The following theorem is the main result of this section.

\begin{theorem}
	\label{thm:matching-sparsifier}
	Let $G=(V,E,w)$ be a weighted graph and $p \in (0,1)$. If the matching polytope of $G$ admits an $\alpha$-balanced contention resolution scheme, then Algorithm~\ref{alg:weighted-matching} is the $(1-O(\eps)) \cdot \max\left\{ \frac 1 2, \left(\frac{1+\alpha e^2}{1+e^2} \right)  \right\} $-approximate  polynomial time sparsifier for the stochastic weighted matching problem $\match$ with sparsification degree $O(1/\epsilon^8 p)$.
\end{theorem}

Before we prove Theorem~\ref{thm:matching-sparsifier}, we first state the following corollary that improves the state-of-the-art degree $\poly(1/p)$ sparsifier for the stochastic weighted matching problem.

\begin{corollary}
	For the stochastic weighted matching problem, there exists a $0.536$-approximate sparsifier with sparsification degree $O(1/p)$.
\end{corollary}
\begin{proof}
	\cite{calum-improved-matching-crs} exhibit a $0.474$-balanced contention resolution scheme for matching. Theorem~\ref{thm:matching-sparsifier} together with $\alpha=0.474$ and sufficiently small $\epsilon>0$ implies the corollary. 
\end{proof}


The proof of Theorem~\ref{thm:matching-sparsifier} relies on $p$ being small. So, before we prove the theorem, in Lemma~\ref{lem:small_prob}, we show that for any $\eps >0$ (constant), without loss of generality we can assume $p \leq \eps^4$. More formally, for any $\eps>0$, we show that any $\alpha$ approximate sparsifier with degree $O(1/p)$ for the class of stochastic weighted matching problems with $p\leq \eps^4$ implies an $\alpha$ approximate sparsifier with degree $O(1/p\eps^4)$ for the same problem class and arbitrary $p \in (0,1)$. The main idea is that given any instance of $\match$, we can split each edge $e$ in the graph into many copies in a way that each copy of the edge $e$ is active with probability $\eps^4\cdot p <\eps^4$ and at least one of them is active with probability exactly $p$. In Lemma~\ref{lem:small_prob}, we show by a simple coupling argument that we can perform such edge splitting with a small loss in the sparsification degree. The proof of
this part is rather technical and, we defer it to
Appendix~\ref{sec:missing_proofs}.

\begin{restatable}[Reduction Lemma]{lemma}{smallprob}
	\label{lem:small_prob}
	If there exists an $\alpha$-approximate sparsifier with degree $\frac{d}{p}$ for the class of stochastic weighted matching problem with $p\leq \epsilon^4$ then there exists an $\alpha$-approximate sparsifier for the same problem class and arbitrary $p \in (0,1)$ with sparsification degree $\frac{d}{p\cdot \epsilon^4}$. 
\end{restatable}

For the rest of the section, we assume that $p\leq \eps^4$. We first define the set of crucial edges and non-crucial edges formally in the following definition.
\begin{definition}\label{def:C-NC}
	Given $\match$, let $q_e$ be the probability of an edge $e$ being in the stochastic optimum solution. We define \emph{crucial edges} as $\ce:=\{e \in E\ :\ q_e \geq \tau(\epsilon)\}$ and \emph{non-crucial edges} as $\nc:=\{e \in E\ :\ q_e < \tau(\epsilon)\}$ where $\tau(\epsilon):= \frac{\eps^3 p }{20\cdot \log \frac 1 \eps}$ is the threshold.
\end{definition}

Given $\match$ and set of crucial and non-crucial edges $\ce$ and $\nc$, we let $\opt_\ce$ and $\opt_\nc$ be the contributions of crucial and non-crucial edges in the stochastic optimum, i.e. $\sum_{e\in \ce} w_e\cdot q_e$ and $\sum_{e\in \nc} w_e\cdot q_e$. Note that $$\opt = \opt_\ce + \opt_\nc.$$

In order to prove Theorem~\ref{thm:matching-sparsifier}, we provide a procedure to construct a matching $M\subseteq Q\cap R$ such that $\E \left[\sum_{i\in M}w_e\right] \geq (1-O(\eps)) \cdot \max\left\{ \frac 1 2, \left(\frac{1+\alpha e^2}{1+e^2} \right)  \right\} \cdot \opt$. Our procedure constructs three matchings $M_{\ce}, M_{\nc}, M_{\aug} \subseteq R\cap Q$ and then picks the matching with the maximum weight. We construct matchings $M_{\ce}, M_{\nc} $ on the queried active crucial and non-crucial edges in $Q_{\greedy}$ similar to the \cite{behnezhad2019stochastic} which satisfies the desired properties described in Lemma~\ref{lem:crucial} and Lemma~\ref{lem:non-crucial}. First, we state that each crucial edge $e\in \ce$ appears in the $Q_{\greedy}$ with probability $1-\eps$ which shows the existence of matching $M_\ce \subseteq Q\cap R \cap \ce$ with expected weight at least $(1-\eps)\cdot \opt_\ce$.

\begin{restatable}[Crucial Edge Lemma \cite{behnezhad2019stochastic}]{lemma}{crucial}
	\label{lem:crucial}
	Given a stochastic weighted matching instance $\match$ and $Q_{\greedy}$ is the set defined in Algorithm~\ref{alg:matching_sparsify}, let $M_\ce$ be the maximum weight matching in the graph $Q_{\greedy}\cap \ce \cap R$, then    
	\begin{equation*}
		\E\left[ \sum_{e\in M_\ce }w_e\right]\geq (1-\eps)\cdot \opt_\ce.
	\end{equation*}
\end{restatable}
Now, following the \cite[Lemma~4.7]{behnezhad2019stochastic}, in Lemma~\ref{lem:non-crucial}, we construct a matching $M_\nc \subseteq R\cap Q_{\greedy} \cap \nc$ on active queried non-crucial edges, such that each $e\in \nc$ is present in $M_\nc$ with probability at least $(1-O(\eps))\cdot q_e$ \footnote{In order to simultneously gurentee $\Pr[e\in M_\nc] \geq (1-\eps) q_e$ for $e\in \nc$, we require a different proof techniques than \cite{behnezhad2019stochastic} to prove the non-crucial lemma.}. We further prove an important property of $M_\nc$ that states that for any non-crucial edge $e\in \nc$, the probability of $e\in M_\nc$ can not decrease when we condition on the events that some of the neighbors of $e$ are inactive. 
\begin{restatable}[Non-Crucial Edges]{lemma}{noncrucial}
	\label{lem:non-crucial}
	Given a stochastic weighted matching instance $\match$, let $Q_{\greedy}$ be the set defined in Algorithm~\ref{alg:matching_sparsify}. There exists a matching $M_\nc \subseteq Q_{\greedy}\cap \nc \cap R$ such that for any non-crucial edge $e\in \nc$, $\Pr[e\in M_{\nc}] \geq (1-12\eps)\cdot q_e$. This implies that,
	\begin{equation*}
		\E\left[ \sum_{e\in M_{\nc} }w_e\right]\geq (1-12\cdot \eps)\cdot \opt_\nc.
	\end{equation*}
	Moreover, for any subset $S \subseteq N(e)$ where $N(e)$ is the set of edges incident to $e$ in graph $G$, we have
	\begin{equation}
		\label{eq:monotone}
		\Pr[e \in M_{\nc} \mid S \cap R = \emptyset] \geq (1-12\cdot\eps)\cdot q_e.
	\end{equation}
\end{restatable}
The proof of the lemma is technically involved and therefore it is delegated to Appendix~\ref{sec:missing_proofs}. Lemma~\ref{lem:crucial} and Lemma~\ref{lem:non-crucial} together implies that
\begin{equation}\label{eq:one-half}
	\E \left[\sum_{e\in \optsol(Q_{\greedy}\cap R)}w_e\right]\geq (1-12\eps) \cdot \max \{\opt_\ce ,\opt_\nc\} \geq (1-12\eps) \cdot 0.5 \cdot \opt.
\end{equation}
First, observe that the $M_\nc$ and $M_\ce$ construct edges on the set of active edges in $Q_{\greedy}$ without considering the edges sampled in $Q_{\crs}$. Note that set $Q_{\crs}$ is sampled independently from the set $Q_{\greedy}$, and includes each edge $e\in Q_{\crs}$ with probability $q_e/p$ independently. Therefore, $Q_{\crs}$ is exactly the output of generic sparsifier discussed in Algorithm~\ref{alg:CRS_to_Sparse} (Section~\ref{sec:CRS_sparse}). Therefore, let $M_{\crs}:=\pi(Q_{\crs} \cap R)$ be the matching constructed by an $\alpha$-balanced contention resolution scheme $\pi$ which ensures that for all $e \in E$, $\Pr[e \in M_\crs] \geq \alpha \cdot q_e$. We refer $M_\crs$ as $\crsM$. 
For any edge $e=(v,u)\in E$, we define  $N(e) = \{e_{1},\dots, e_{k}\}$ as the incident edges on vertices $u$ and $v$ in the graph $G$ or set of neighbours of an edge $e$.

We obtain improvement on the approximation ratio from $0.5$ by constructing a third matching on the set of edges $Q_\crs \cup (Q_\greedy \cap \nc) \cup R$ via augmenting the matching $M_\nc$ with the matching $\crsM$. Our augmentation simply adds a non-crucial edge $e\in M_\nc$ to $\crsM$ if both endpoints of the edge $e$ are unmatched in $\crsM$. Algorithm~\ref{alg:query-matching} describes our augmentation procedure in detail. 

Our key observation is that any non-crucial edge $e\in \nc$ has a small probability of being sampled in the set $Q_{\crs}$. However, with some non-trivial probability, both endpoints of the edge $e$ will be unmatched in $\crsM$. On the other hand, \eqref{eq:monotone} property of $M_\nc$ from Lemma~\ref{lem:non-crucial} implies that when a non-crucial edge $e\notin Q_{\crs}$ and both endpoints of the edge $e$ are unmatched in $\crsM$, then we are more likely to add the edge $e$ to $M_\nc$. Hence, we can improve the weight of $\crsM$ with non-trivial probability.

\begin{algorithm}[H]
	\caption{Construction of the matching $M_\aug$ on $Q\cap R$ \label{alg:query-matching}}
	\begin{algorithmic}[1]
		\State Let $M_\nc$ be the matching on $Q_\greedy \cap R\cap \nc$ satisfying property of stated Lemma~\ref{lem:non-crucial}.
		\State Let $M_\crs \gets \pi(Q_\crs \cap R)$ be the matching produced by $\alpha$-balanced truncated CRS $\pi$.
		\State $M_\aug \gets M_\crs$.
		\For {all edges $e \in M_{\nc}$}
		\If{$M_{\crs} \cup {e}$ is a matching}
		\State Update $M_\aug \gets M_\aug \cup \{e\}$.
		\EndIf
		\EndFor
		\State \textbf{Output}: $M_\aug$.
	\end{algorithmic}
\end{algorithm}

More formally, first we show that for any non-crucial edge $e:=(u,v)\in \nc$, both endpoints of $e$ are unmatched in the $\crsM$  with probability at least $1/e^2$. Later, we use property \eqref{eq:monotone} of $M_\nc$ from Lemma~\ref{lem:non-crucial} to guarantee that when a non-crucial edge $e \notin Q_{\crs}$ and both endpoints of $e$ are unmatched in $\crsM$, we can guarantee that $e\in M_\nc$ with probability approximately $q_e$. Therefore, we can add such a non-crucial edge $e$ to $\crsM$ with probability approximately $\frac {q_e} {e^2}$ 

We show that for any non-crucial edge $e=(u,v)\in \nc$, the number of incident edges on vertices $u$ and $v$ in the set $Q_{\crs}$ are concentrated around $\frac 1 p$ with high probability due to $p$ being sufficiently small. Such a property ensures that if a non-crucial edge $e\in \nc$ is included in the set $Q_{\crs}$ then there are not many neighbors of $e$ in the set $Q_{\crs}$ with high probability. Thus, if all these neighbors are inactive, we guarantee that both endpoints of $e$ are unmatched in $\crsM$. Formally, we define $\mathcal E_e :=|Q_{\crs}\cap N(e)| \leq \frac{(1+\eps)\cdot 2}{p}$ and show that if a non-crucial edge $e=(u,v)$ is not included in the $Q_{\crs}$ then with high probability the event $\mathcal E_e$ happens, i.e. there are not many edges incident to the endpoints of the edge $e$.

\begin{restatable}{proposition}{concentration}\label{prop:high_prob_events}
	Given instance of $\match$, for any non-crucial edge $e=(u,v)\in \nc$, 
	\begin{equation*}
		\Pr[\mathcal E_e \mid e\notin  Q_{\crs}] \geq 1-2\epsilon.
	\end{equation*}    
\end{restatable}
The proof of Proposition~\ref{prop:high_prob_events} is delegated to Appendix~\ref{sec:missing_proofs}. In the next lemma, we show that whenever a non-crucial edge $e\notin Q_{\crs} $, both endpoints of the edge $e$ are unmatched in $\crsM$ with probability at least $1/e^2$. For any non-crucial edge $e=(v,u)\in \nc$ with $N(e) = \{e_{1},\dots, e_{k}\}$ being the incident edges on vertices $u$ and $v$ in the graph $G$, we  define the event
\begin{equation*}
	\mathcal F_e := Q_{\crs}\cap R \cap N(e) = \emptyset.    
\end{equation*}
\begin{restatable}{lemma}{feasibilityconst}
	\label{lem:feasibility_const}
	Given a stochastic weighted matching problem $\match$, for any non-crucial edge $e=(u,v)\in \nc$ and any $\eps>0$,
	$$ \Pr[\text{$u$ and $v$ are unmatched in }\crsM\mid e\notin Q_{\crs} ] \geq \Pr[\mathcal F_e\mid e\notin Q_{\crs}] \geq \frac{1-4\eps}{e^2} - \eps^2.$$
\end{restatable}
\begin{proof}
	
	First, notice that,
	\begin{equation*}
		\Pr[\text{$u$ and $v$ are unmatched in }\crsM\mid e\notin Q_{\crs} ] \geq \Pr[\mathcal F_e\mid e\notin Q_{\crs}],
	\end{equation*}
	because if all edges in $Q_{\crs}\cap N(e)$ are inactive then both endpoints $u$ and $v$ are unmatched in $\crsM$. Now, we can express,
	\begin{align*}
		\Pr[\mathcal F_e\mid e\notin Q_{\crs}] &\geq \Pr[\mathcal F_e\mid e\notin Q_{\crs}\land \mathcal E_e]\cdot \Pr[\mathcal E_e\mid e\notin Q_{\crs}]\\
		&\geq (1-2\eps)\cdot  \Pr[\mathcal F_e\mid e\notin Q_{\crs}\land \mathcal E_e].
	\end{align*}
	By observing 
	\begin{align*}
		\Pr[\mathcal F_e\mid \mathcal E_e] \leq  \Pr[\mathcal F_e\mid \mathcal E_e \land e\notin Q_{\crs}] + \Pr[e\in Q_{\crs}]
	\end{align*}
	we obtain
	$$ \Pr[\mathcal F_e\mid \mathcal E_e \land e\notin Q_{\crs}]\geq \Pr[\mathcal F_e\mid \mathcal E_e] - \Pr[e\in Q_{\crs}] \geq \Pr[\mathcal F_e\mid \mathcal E_e] - \eps^2.$$
	
	Combining everything, we obtain, 
	\begin{align*}
		\Pr[\mathcal F_e\mid e\notin Q_{\crs}] &\geq (1-2\eps)\cdot  (\Pr[\mathcal F_e\mid \mathcal E_e] - \eps^2)\\
		&\geq (1-2\eps) \left((1-p)^\frac{2}{p(1+\eps)} - \eps^2\right) \geq (1-2\eps) \left((1-p^2)e^{-\frac 2{1+\eps}} -\eps^2 \right)\\
		&\geq (1-2\eps) \left((1-p^2)(1-\eps)e^{-2} -\eps^2 \right)\\
		&\geq \frac{1-4\eps}{e^2} -\eps^2
	\end{align*}
	and conclude the proof.
\end{proof}

The previous lemma provides a lower bound on the probability of both endpoints of a non-crucial edge $e\in \nc$ being unmatched in $\crsM$. Next, we prove a lower bound on each edge $e$ being selected in the matching $M_\aug$. First, observe that the $M_\aug$ contains $\crsM$ therefore, each edge $e\in E$ is selected in $M_\aug$ with probability at least $\alpha\cdot q_e$. Then any non-crucial edge $e\in \nc$ that is not selected in $Q_\crs$ (hence it is not selected in $\crsM$) is further added in the $M_\aug$ when event $\mathcal F_e$ occurs. First, note that the event $\mathcal F_e$ implies $Q_\crs \cap N(e) \subseteq N(e)$ are inactive. Hence, given $\mathcal F_e$, property \eqref{eq:monotone} from Lemma~\ref{lem:non-crucial} further implies that $e\in \nc$ is selected in $M_\nc$ with probability approximately $q_e$. Formalizing the above discussion, we obtain the following Lemma.

\begin{lemma}\label{lem:augment}
	Let $M_\aug$ be the output of the procedure described in Algorithm~\ref{alg:query-matching} then,
	$$ \Pr[e \in M_\aug] \geq q_e \cdot \alpha\ \forall e \in \ce \qquad \text{and}\qquad \Pr[e \in M_\aug] \geq q_e \cdot \left(\alpha + \frac{1-O(\epsilon)}{e^2}\right)\ \forall e \in \nc.$$
\end{lemma}
\begin{proof}
	We first show some probabilistic guarantees of $M_\aug$ for crucial and non-crucial edges separately. Since $\pi(Q_{\crs} \cap R)=M_\crs  \subseteq M_\aug$ is constructed by $\alpha$-balanced CRS $\pi$, we have $\Pr[e \in M_\aug] \geq \alpha \cdot q_e$ for all $e \in E$. For a non-crucial edge $e \in \nc$, we can improve this lower-bound as follows. For any non-crucial edge $e \in \nc$
	
	\begin{align*}
		\Pr[e \in M_\aug] &\geq \Pr[e \in M_\crs] + \Pr[e \notin M_\crs \wedge e \in M_\nc \wedge \mathcal F_e]\\
		&\geq \Pr[e \in M_\crs] + \Pr[e \notin Q_{\crs} \wedge e \in M_\nc \wedge \mathcal F_e] && (M_\crs \subseteq Q_\crs)\\
		&\geq \alpha \cdot q_e + \Pr[e \in M_\nc \wedge \mathcal F_e \mid e \notin Q_{\crs}] \cdot \Pr[e \notin Q_{\crs}] \\
		&\geq \alpha \cdot q_e + (1-\epsilon) \cdot \Pr[e \in M_\nc \wedge \mathcal F_e \mid e \notin Q_{\crs}], && (q_e/p \leq \epsilon)
	\end{align*}
	Where the first inequality holds because $\mathcal F_e$ implies that both endpoints of $e$ are unmatched in $\crsM$. Next, we further compute the conditional probability as follows
	\begin{align*}
		\Pr[e \in M_\nc \wedge \mathcal F_e \mid e \notin Q_{\crs}] &= \Pr[e \in M_\nc \mid \mathcal F_e \wedge e \notin Q_{\crs}] \cdot \Pr[\mathcal F_e \mid e \notin Q_{\crs}]\\
		&\geq  \Pr[e \in M_\nc \mid N(e) \cap Q_{\crs} \cap R = \emptyset] \cdot \Pr[\mathcal F_e \mid e \notin Q_{\crs}]\\
		&\geq (1-12\epsilon) \cdot q_e \cdot \Pr[\mathcal F_e \mid e \notin Q_{\crs}] && (\text{\eqref{eq:monotone}  from Lemma~\ref{lem:non-crucial}})\\
		&\geq (1-12\epsilon) \cdot q_e \cdot \left(\frac{1-4\eps}{e^2}-\eps^2\right) && (\text{Lemma~\ref{lem:feasibility_const}})\\
		&\geq (1-O(\epsilon)) \cdot \frac{q_e}{e^2}.
	\end{align*}
	Thus, we summarize lower bounds as follows
	$$ \Pr[e \in M_\aug] \geq q_e \cdot \alpha\ \forall e \in \ce \qquad \text{and}\qquad \Pr[e \in M_\aug] \geq q_e \cdot \left(\alpha + \frac{1-O(\epsilon)}{e^2}\right)\ \forall e \in \nc.$$
\end{proof}

Now, we are ready to prove the main theorem of this section.
\begin{proof}[Proof of Theorem~\ref{thm:matching-sparsifier}]
	Let $Q$ be the output of Algorithm~\ref{alg:matching_sparsify}. Let $M_\ce$ be the matching defined in Lemma~\ref{lem:crucial}, $M_\nc$ be the matching constructed in Lemma~\ref{lem:non-crucial} and $M_\aug$ be the matching constructed by Algorithm~\ref{alg:query-matching}. Notice that $M_\aug,M_\nc, M_\ce \subseteq Q \cap R$. 
	Recall that $\opt_\ce$ and $\opt_\nc$ are expected contributions of crucial and non-crucial edges to the optimum value. First, by Lemma~\ref{lem:augment}, we obtain
	\begin{align*} 
		\Ex[f(M_\aug)] &\geq \alpha \cdot \sum_{e\in \ce} w_e\cdot q_e + \left(\alpha + \frac{1-13\epsilon}{e^2}\right) \cdot \sum_{e\in \nc} w_e\cdot q_e \\
		&\geq \alpha \cdot \opt_\ce + \left(\alpha + \frac{1-13\epsilon}{e^2}\right) \cdot \opt_\nc.
	\end{align*}
	By Lemma~\ref{lem:crucial} and Lemma~\ref{lem:non-crucial} we also know that
	$$ \Ex[f(M_\ce)] \geq (1-\epsilon) \cdot \opt_\ce \qquad \text{and} \qquad \Ex[f(M_\nc)] \geq (1-\epsilon) \cdot \opt_\nc.$$
	So, combining three inequalities will give us 
	\begin{align*}
		\Ex[f(\optsol(Q \cap R))] &\geq \Ex[\max\{f(M_\ce),f(M_\aug),f(M_\nc)\}]\\
		&\geq \max\{\alpha \cdot \opt_{\ce} + \left(\alpha + \frac{1-13\epsilon}{e^2}\right) \cdot \opt_{\nc},\ (1-\epsilon) \cdot \opt_{\ce},\ (1-\epsilon) \cdot \opt_{\nc}\}.
	\end{align*}
	Minimizing the right hand side for $\opt_\ce$ and $\opt_\nc$ we obtain the following guarantee
	$$ \Ex[f(\optsol(Q \cap R))] \geq (1-13\eps) \cdot \frac{1+\alpha\cdot e^2}{1+e^2}\cdot \opt.$$
	Combining with Equation~\ref{eq:one-half}, we conclude the proof of the theorem.
\end{proof}

\section{Additive Optimization over the Intersection of $k$ Matroids}\label{sec:sparse_matroidintersections}

Given our $(1-\epsilon)$-approximate sparsifier for additive optimization over a single matroid constraint, a natural question is whether the natural generalization of this algorithm to the intersection of matroids is $(1-\epsilon)$-approximate. This turns out to not be the case even for bipartite matching (the intersection of two partition matroids): the algorithm which iteratively selects a maximum-weight matching and removes it from the graph, for a total of $\poly(1/p)$ iterations, has an approximation guarantee bounded away from $1$ as shown by \cite{blum2015ignorance}.


It is well known that the greedy algorithm is $\frac 1 k$-approximate for maximizing an additive function over the intersection of $k$ matroids. We conjecture that this greedy algorithm can be converted to a $\frac{1-\epsilon}{k}$-approximate sparsifier with degree $\poly\left(\frac{1}{p\epsilon}\right)$. However, the main challenge here is that, unlike for a single matroid, multiple solutions for matroid intersection do not always ``combine'' well.   In this section, we prove a slightly weaker sparsification result for additive optimization over the intersection of $k$ matroid constraints, which nevertheless beats the best known bound of $1/(k+1)$ on the correlation gap of $k$-matroid intersection (see \cite{adamczyk2018random}), and therefore outperforms our generic sparsifier for this problem (see Corollary~\ref{coro:crsspars}).


We exhibit a $\frac{(1-\epsilon)}{k+\frac{1}{k+1}}$-approximate sparsifier for additive optimization over the intersection of $k$~matroids. The following theorem is the main result of this section.

\begin{theorem}\label{thm:matroid_intersection_sparsifier}
 For each $\epsilon>0$, there is a $\frac{(1-\epsilon)}{k+\frac{1}{k+1}}$-approximate sparsifier of degree $O \left(\frac 1 {\eps\cdot p} \log \frac 1 \epsilon \right)$ for stochastic packing problem $\spack$ when $(E,\I)$ is the intersection of $k$ matroids and $f$ is additive.
\end{theorem}
We describe our sparsifier in Algorithm~\ref{alg:Sparse_MI}. This sparsifier samples $Q_1,\dots, Q_\tau \sim_{iid} \sopt$, where $\sopt$ is a stochastic optimal oracle, and outputs $Q = \bigcup_{i=1}^\tau Q_i $. Similar algorithms with degree $\poly\left( \frac 1 p\right)$ have been considered for the sparsification of matching by \cite{behnezhad2019stochastic,behnezhad2020unweighted}, and were shown to be $0.6568$-approximate for the unweighted additive objective and $0.501$-approximate for a general additive objective.
The same style of algorithm would later yield a $(1-\epsilon)$-approximate sparsifier for the unweighted additive objective, but with degree $exp(exp(exp(1/p)))$ \cite{behnezhad2020unweighted}.



\begin{algorithm}[H]
    \caption{Sparsifier for additive optimization over the intersection of $k$ matroid constraints.}\label{alg:Sparse_MI}
   \textbf{Input: }$\spack$ where $(E,\I)$ is the intersection of $k$ matroids and $f$ is additive \\
    \textbf{Input: }$\sopt$ (stochastic optimum oracle for $\spack$)\\
    \textbf{Input: }$\epsilon>0$
    \begin{algorithmic}
    \State $\tau \gets \frac{2}{\eps p}\log \frac{2}{\epsilon}$
    
    \State $Q\gets \emptyset $
  \For{all $i=1\dots,\tau$}
    \State Sample $Q_i \sim \sopt$ 
    \State $Q\gets Q\cup Q_i$
    \EndFor
    \State \textbf{Output: }Sparse subset $Q$. 
    \end{algorithmic}
\end{algorithm}


In order to prove Theorem~\ref{thm:matroid_intersection_sparsifier}, we provide a procedure for constructing a feasible solution $I\subseteq Q\cap R$ such that $\mathbb E [\sum_{i\in I}w_e] \geq \frac{(1-\epsilon)}{k+1/(k+1)} \opt$. First, we recall the following well-known lemma.

\begin{lemma}[See e.g. \cite{welsh10}]\label{lem:exchange_property}
    Let $S_1,S_2$ be any two independent sets of matroid $\M$, if $e\in S_1\setminus S_2$ and  $e\in \textbf{Span}(S_2)$ then $\exists f\in S_2\setminus S_1$ such that $S_1\setminus e \cup f\in \I$ and $S_2\setminus f\cup e \in \I$.
\end{lemma}
  


\begin{algorithm}[H]
    \caption{Procedure for Constructing Feasible $T$ for Lemma~\ref{lem:exchangemap}}\label{alg:constructT}
   \textbf{Input:} Matroid $\M=(E,\I)$ \\
      \textbf{Input:} Independent sets $S_1,S_2$ of $\M$ \\
    \textbf{Input:} Set $R \sse E$ of active elements 
    \begin{algorithmic}
    \State Let $S_1\setminus S_2 =\{e_1,\dots,e_\ell\}$.
    \For{all $i=1,\dots,\ell$}
        \If{$e_i\notin \textbf{Span}(S_2)$}
            \If{$e_i\in R$}
                \State $S_2\gets S_2\cup e_i$ 

            \Else
                \State $S_1\gets S_1\setminus e_i$                 
            \EndIf
        \Else
            \State Let $f_i\in S_2\setminus S_1$ be as in Lemma~\ref{lem:exchange_property}
            \If{$e_i\in R$}
                \State $S_2\gets S_2\setminus f_i \cup e_i$ 
            \Else
                \State $S_1\gets S_1\setminus e_i \cup f_i$ 
            \EndIf
        \EndIf
    \EndFor
    \State  \textbf{Output: }$T=S_2$
    \end{algorithmic}
\end{algorithm}

Let $S_1$ and $S_2$ be two independent sets of the same matroid $\M=(E,\I)$. Let $R\subseteq E$ be the (random) set of active elements with parameter $p$. Algorithm~\ref{alg:constructT} is a procedure which augments the active elements in $S_1$, i.e. $S_1\cap R$, with elements of $S_2$. We show in Lemma~\ref{lem:exchangemap} that this augmentation procedure includes each element of $S_2$ with probability $1-p$,\footnote{We mention that similar guarantees were claimed by \cite{maehara2019submodular} for their ``uniform exchange maps''. We believe  there is an irreparable bug in their analysis. Therefore, we give an alternative and constructive proof for the existence of such maps here.} and therefore outputs an independent set $T$ with $f(T) \geq \Ex[f(S_1 \cap R)] + (1-p) \cdot f(S_2)$. 


\begin{lemma}\label{lem:exchangemap}
Let $\M=(E,\I)$ be a matroid, and let $S_1$ and $S_2$ be two independent sets $\M$. Let $R\subseteq E$ include each element of $E$ independently with probability $p$. Algorithm~\ref{alg:constructT} outputs a random set $T\subseteq (S_1\cap R) \cup S_2$ satisfying:	

\begin{enumerate}
	\item $S_1\cap S_2 \subseteq T$ with probability $1$.
        \item $T\in \I$ with probability $1$. 
	\item $(S_1\setminus S_2)\cap R \subseteq T$, i.e. $\Pr[e\in T] = p$ for all $e\in S_1\setminus S_2$.
	\item $\Pr[f\in T]\geq 1-p$ for all $f\in S_2\setminus S_1$. 
\end{enumerate}	
\end{lemma}
\begin{proof}
The algorithm only modifies (adds or deletes) items from the symmetric difference of $S_1$ and $S_2$, therefore, $S_1 \cap S_2 \subseteq T$ and property 1 holds. Observe that if $e_i \notin \spn(S_2)$ then $S_2 \cup e_i \in \I$. If $e_i \in \spn(S_2)$, by Lemma~\ref{lem:exchange_property}, we can find $f_i \in S_2 \setminus S_1$ which can be exchanged with $e_i$ without violating feasibility. Therefore $S_1$ and $S_2$ remains feasible throughout the procedure, and so property 2 holds. Observe that we add $e_i \in S_1 \setminus S_2$ to $S_2$ if and only if $e_i \in R$, which establishes property~3. Therefore, for all elements $e_i\in S_1\setminus S_2$ we have $\Pr[e_i\in T]=p$.
    
     We say that an element $f_i \in S_2 \setminus S_1$ competes with $e_i\in S_1\setminus S_2$ if the algorithm tries to swap $e_i$ and $f_i$.  We say that $f_i$ ``wins" this competition if and only if $e_i$ is inactive. If $f_i$ wins then it remains in $S_2$ and it is added to $S_1$ as well; therefore, $f_i$ remains in $S_2$ until the end of algorithm. In other words, $f_i$ competes for its position at most once.  It follows that $f_i \in S_2 \setminus S_1$ appears in the output $T$ with probability at least $1-p$, and so property 4 holds.
\end{proof}

Now, if $S_1$ and $S_2$ are independent in all $k$ matroids, i.e. $S_1, S_2 \in \bigcap_{\ell = 1}^k \I_\ell$, then we can run the procedure described in Algorithm~\ref{alg:constructT} for each matroid $\M_\ell$ for $\ell=1,\dots, k$ and obtain sets $T(\ell)\in I_\ell$ satisfying all properties of Lemma~\ref{lem:exchangemap} for their respective matroids. We then get the following corollary for their intersection. 
 
\begin{corollary}\label{coro:exchangemap}
    Let $\M_1, \ldots \M_k$ be matroids with $\M_\ell=(E,\I_\ell)$, and let $\I = \bigcap_{\ell=1}^k \I_\ell$ be their common independent sets. Let $S_1$ and $S_2$ be in $\I$.  Let $R\subseteq E$ include each element of $E$ independently with probability $p$. Let $T(\ell) \in \I_\ell$ be the output of Algorithm~\ref{alg:constructT} for matroid $\M_\ell$, for each $\ell \in [k]$. The set $T:=\bigcap_{\ell=1}^kT(\ell)$ satisfies:
\begin{enumerate}

  \item $S_1\cap S_2 \subseteq T$ with probability $1$.
  \item $T\in \I$ with probability $1$. 
  \item $(S_1\setminus S_2)\cap R \subseteq T$, i.e. $\Pr[e\in T] = p$ for all $e\in S_1\setminus S_2$. 
  \item $\Pr[f\in T]\geq (1-p)^k$ for all $f\in S_2\setminus S_1$
\end{enumerate}
\end{corollary}

Now we are ready to state Algorithm~\ref{alg:constructI} that finds an active feasible subset of $Q \cap R$ by using Corollary~\ref{coro:exchangemap}, where $Q$ is the output of Algorithm~\ref{alg:Sparse_MI}.

Algorithm~\ref{alg:constructI} visits samples $Q_1,\dots, Q_\tau$ in order. The algorithm starts with an empty ``working solution" $I(0)$ at the beginning. At the first iteration of the outer for-loop, Algorithm~\ref{alg:constructI} adds all the active elements from $Q_1$ to the working solution $I(1) = Q_1\cap R$. Moreover, in the inner for-loop, the algorithm iterates over all $j>1$ and updates $Q_j$. In particular, $I(1)$ is added to $Q_j$ and some elements are removed to preserve feasibility, as per Corollary~ \ref{coro:exchangemap}. This process is then repeated for $t=2,\dots,\tau$.

\begin{algorithm}[H]
    \caption{Procedure for Constructing $I\subseteq Q\cap R$, feasible for $\bigcap_{\ell=1}^k \M_\ell$}\label{alg:constructI}
    \textbf{Input: } $Q= \bigcup_{t=1}^\tau Q_t$ as in Algorithm~\ref{alg:Sparse_MI}, and the set $R$ of active elements.
    \begin{algorithmic}
    \State $I(0)\gets \emptyset$
    \For{all $t=1,\dots,\tau$}
        \State $I(t)\gets (Q_t\cap R)$.
        \For{all $i=t+1, \dots, \tau$}
            \State Let $T_\ell$ be the output of Algorithm \ref{alg:constructT} for $\M_\ell$ with inputs $S_1=Q_t$, $S_2=Q_i$, and $R$.
            \State Update $Q_i \gets \bigcap_{\ell=1}^k T_\ell$.
        \EndFor
    \EndFor
    \State \textbf{Output:} $I(\tau)$.
    \end{algorithmic}
\end{algorithm}
First, we show that Algorithm~\ref{alg:constructI} outputs a feasible set of active elements.

\begin{proposition}
    The output $I^*=I(\tau)$ of Algorithm~\ref{alg:constructI} is feasible for $\bigcap_{\ell=1}^k \M_{\ell}$, and satisfies  $I^*\subseteq R\cap Q$, with probability $1$.
\end{proposition}
\begin{proof}

Observe that in each step $i$ of the inner loop, the algorithm updates $Q_i$ by using the procedure defined by Corollary~\ref{coro:exchangemap}. Therefore, all $Q_i$'s remain feasible. Moreover, the output $I(\tau) = Q_\tau \cap R$ at the end of the procedure only contains active elements. Therefore Algorithm~\ref{alg:constructT} outputs a feasible subset of $Q \cap R$.
\end{proof}

The following lemma proves a lower bound on the probability that Algorithm~\ref{alg:constructI} includes each element of $Q$. 

\begin{lemma}  \label{lem:exist-guarantee}
	Let $I^* = I(\tau)$ be the output of Algorithm~\ref{alg:constructI}. For any $e \in E$, we have 
    $$\Pr[e \in I^* \mid e \in Q_i \setminus Q_{1:i-1}] \geq p \cdot (1-p)^{k(i-1)}$$ 
    where $Q_{1:i-1} := \bigcup_{\ell=1}^{i-1} Q_i$.
\end{lemma}

\begin{proof}
    First of all, we fix $e \in E$ and condition on $e\in Q_i \setminus Q_{1:i-1}$. Let $Q_i(t)$ be the set $Q_i$ after step $t$ of the outer loop for $t=0,\dots, i$. For simplicity of notation, denote $Q_i=Q_i(0)$. It suffices to  show that 
    $$\Pr[e \in I(i)] \geq p \cdot (1-p)^{k \cdot (i-1)}$$
    since if $e \in I(i)$, then $e \in Q_j(i)$ for any $j \geq i$, and so $e \in I(\tau)$.

    We now lowerbound the probability that $e \in I(i)$. We have
    $$ \Pr[e \in I(i)] = \Pr[e \in Q_i(i-1) \cap R] = p \cdot \Pr[e \in Q_i(i-1)]$$
    where the last equality follows from the principal of deferred decisions and the fact that $e_i\in Q_i\setminus Q_{1:i-1}$. Now, iterative application of Corollary~\ref{coro:exchangemap} completes the proof:
    \begin{align*}
        \Pr[e \in Q_i(i-1)] &= \prod_{t=1}^{i-1} \Pr[e \in Q_i(t) \mid e \in Q_i(t-1)] \\ &\geq  \prod_{t=1}^{i-1} (1-p)^k = (1-p)^{k\cdot(i-1)}.
    \end{align*}
\end{proof}

Now, we are ready to prove Theorem~\ref{thm:matroid_intersection_sparsifier}.
\begin{proof}[Proof of Theorem~\ref{thm:matroid_intersection_sparsifier}]
  Let $\set{w_e}_{e \in E}$ be the weights associated with the additive objective $f$. Let $Q_1,\dots, Q_\tau \sim_{iid} \sopt$ be the samples selected by Algorithm~\ref{alg:Sparse_MI}, $Q = \bigcup_{i=1}^\tau Q_i$ be their union, and $R$ be the set of active elements. Let $q_e = \Pr_{\O\sim \sopt}[e\in \O]$ be the probability that $e$ is in the stochastic optimum solution $\sopt$, for all $e\in E$.

  Given $\{q_e\}_{e\in E}$, We say that element $e\in E$ is crucial if $q_e\geq  p\cdot \eps / 2$, otherwise we call $e$ non-crucial. We denote the set of crucial elements by $\operatorname{C}$, and the set of non-crucial elements by $\operatorname{NC}$. We can write the optimum value of $\spack$ as follows:
\begin{equation}
    \opt = \sum_{e\in E}w_e\cdot q_e = \sum_{e\in \operatorname{C}}w_e\cdot q_e +  \sum_{e\in \operatorname{NC}}w_e\cdot q_e
\end{equation}
Observe that for all $e\in E$, $\Pr[e\in Q_i] = q_e\leq p$ and that event is independent for $i=1,\dots , \tau$. So, $\Pr[e\notin Q] = (1-q_e)^\tau $. For any crucial element $e\in \operatorname{C}$, we have
\begin{equation*}
   \Pr[e\in Q] = 1-(1-q_e)^\tau \geq 1-\left(1-\frac{\eps\cdot p}{2}\right)^\tau  \geq 1-\epsilon.
\end{equation*}
Therefore, we obtain our first bound on the performance of Algorithm~\ref{alg:Sparse_MI} as follows, where $\O \sse R$ is a stochastic optimum solution with distribution $\sopt$.
\begin{equation}
   \E\left[\max_{\substack{I\subseteq Q\cap R\\ I\in \I}} f(I) \right]\geq \E [f(\O\cap Q)]  \geq (1-\epsilon)\cdot \sum_{e\in \operatorname{C}} w_e\cdot q_e \label{eq:value_important},
 \end{equation}

Now let $I^* \subseteq Q \intersect R$ be the output of Algorithm~\ref{alg:constructI}. We can bound the probability that each element is in $I^*$ as follows:

\begin{align*}
    \Pr[e\in I^*] &= \sum_{i=1}^\tau \Pr[e\in Q_i \setminus Q_{1:i-1}]\cdot \Pr[e\in I^* \mid e\in Q_i \setminus Q_{1:i-1}]\\
    &= \sum_{i=1}^\tau \Pr\left[e\in Q_i \wedge \bigcap_{j=1}^{i-1}\{e\notin Q_j\} \right] \cdot \Pr[e\in I^* \mid e\in Q_i \setminus Q_{1:i-1}]\\
    &\geq \sum_{i=1}^\tau q_e (1-q_e)^{i-1} \cdot p (1-p)^{k \cdot (i-1)} && (\text{Lemma~\ref{lem:exist-guarantee}})\\
                  & \geq  p\cdot q_e \cdot \frac{1-((1-p)^{k}\cdot(1-q_e))^{\tau}}{1-(1-p)^k\cdot (1-q_e)} && (\text{Truncated geometric sum})\\
    &\geq \left(1-\frac{\eps}{2}\right) \cdot \frac{p \cdot q_e}{1-(1-p)^k\cdot (1-q_e)} && (\tau \geq \frac{1}{p}\log\left(\frac{2}{\epsilon}\right) \text{ and } (1-q_e)\leq 1)\\
                  &\geq \left(1-\frac{\eps}{2}\right) \cdot \frac{p \cdot q_e}{1-(1-kp)\cdot (1-q_e)} && ((1-p)^k \geq 1-kp)\\
                  &= \left(1-\frac{\eps}{2}\right) \cdot \frac{p \cdot q_e}{kp+q_e - kpq_e} && ((1-p)^k \geq 1-kp)\\
    &\geq  \left(1-\frac{\epsilon}{2}\right)\cdot \frac{q_e}{k+q_e/p}
\end{align*}
Since $q_e \leq p$, we have  $\Pr[e\in I^*]\geq \frac{1- \eps}{k+1} \cdot q_e$ for all elements $e$, in particular crucial elements. For non-crucial elements $e$, we we obtain the following tighter bound:

\begin{equation*}
\Pr[e\in I^*]\geq \left(1-\frac{\eps}{2}\right) \cdot \frac{q_e}{k+\eps/2}   \geq \frac{(1-\eps)}{k}  \cdot q_e 
\end{equation*}
 We thus obtain our second bound on the performance of Algorithm~\ref{alg:Sparse_MI}.
\begin{align}
   \E\left[\max_{\substack{I\subseteq Q\cap R\\I\in \I}} \sum_{e\in I}w_e \right] \geq \E \left[\sum_{e\in I^*}w_e\right]
   \geq \frac{(1-\epsilon)}{k+1}\cdot \sum_{e\in \operatorname{C}}w_e\cdot q_e + \frac{(1-\epsilon)}{k}\cdot \sum_{e\in \operatorname{NC}}w_e\cdot q_e \label{eq:value_I}
\end{align}

Combining our bounds from Equations~\eqref{eq:value_important} and~\eqref{eq:value_I}, we get
\begin{align*}
   \E\left[\max_{\substack{I\subseteq Q\cap R\\I\in \I}} \sum_{e\in I}w_e \right] &\geq \max\left\{(1-\epsilon)\cdot \sum_{e\in \operatorname{C}}w_e\cdot q_e, \quad \frac{(1-\epsilon)}{k+1}\cdot \sum_{e\in \operatorname{C}}w_e\cdot q_e + \frac{(1-\epsilon)}{k}\cdot \sum_{e\in \operatorname{NC}}w_e\cdot q_e \right\}\\
   &\geq \frac{(1-\epsilon )}{k+1/(k+1)}\cdot \sum_{e\in E}w_e\cdot q_e\\
   &\geq \frac{(1-\epsilon )}{k+1/(k+1)}\cdot \opt,
\end{align*}
concluding the proof of the theorem. 
\end{proof}


\section{Submodular Optimization}
\label{sec:submodular}

In this section, we consider the sparsification of stochastic packing problems with a monotone submodular objective. First, we design improved polynomial-time sparsifiers for some SPPs with a coverage submodular objective. Second, we show an information-theoretic impossibility result for sparsification in this setting. In particular, even for optimizing a coverage function subject to a uniform matroid constraint we show that no sparsifier with degree independent of the number of elements can achieve an approximation ratio better than $\left( 1-\frac 1 e\right)$.

\subsection{Coverage Function Optimization} \label{sec:coverage}
 A set function $f:2^E\rightarrow \R_+ $ is a \emph{coverage function} if elements of the ground set $i \in E$ index subsets $A_i$ of some universe $U$, and  $f(S)= \left|\bigcup_{i \in S} A_i \right|$ for $S \subseteq E$. Whereas in general $U$ may be infinite, or even an arbitrary measure space, here we consider the case where $U$ is finite. For clarity, we refer to $i \in E$ as \emph{elements} and $j \in U$ as \emph{points}.

In this section, we consider SPPs $\spack$ when $f$ is a coverage function. For polynomial-time implementation we also require that $f$ is given explicitly, and that the polytope of the set system admits an efficient separation oracle (which is the case for matroids and matroid intersections). We employ monotone contention resolution schemes for the $(E,\I)$ as a proof tool, and express our approximation ratio as a function the best balance ratio of such a monotone CRS. The most notable instantiation of this result is for submodular optimization subject to a matroid constraint: we improve the $\left(1-\frac{1}{e}\right)^3$-approximate polynomial-time sparsifier of Corollary~\ref{coro:crsspars} to $\left(1-\frac{1}{e}\right)^2$ when the objective is a coverage function, while maintaining the same degree of $\frac{1}{p}$.


     Our sparsifier is based on rounding the following LP relaxation of the stochastic optimum of of $\spack$ to a sparse subsect of degree $\frac 1 p$. Here $P_\I = \operatorname{convexhull}\{\mathbbm{1}_I: I\in \I\} \sse [0,1]^E$ is the polytope of the set system $(E,\I)$.

\begin{align*}
    \max\qquad &\sum_{j \in U} y_j\\
    \text{subject to}\qquad &y_j \leq \sum_{i: j \in A_i} x_i &\text{for all } j \in U,\\
    &y_j \leq 1 &\text{for all }j \in U, \tag{LP}\label{lp:coverage}\\
    &x_i \leq p &\text{for all }i \in E,\\
    &\vec{x} \in P_\I
\end{align*}

It is easy to see that \eqref{lp:coverage} is a relaxation of the stochastic optimum, therefore its optimal value $\opt_{LP}$ is at least $\opt$. Moreover, this LP can be solved in polynomial time relative to a separation oracle for $P_\I$.

Given an optimal solution $(\vec{x},\vec{y})$ to \eqref{lp:coverage}, we construct our sparse set $Q \sse E$ as follows: we include each element $i$ independently with probability $\frac{x_i}{p}$. Notice that the expected cardinality of $Q$ is at most $\frac{\rank((E,I)))}{p}$, and hence our sparsifier has degree $\frac 1 p$.

Next, we bound the approximation ratio of our sparsifier. First we compute the probability of each point $j \in U$ being covered by $Q\cap R$ as follows: 

\begin{align*}
   \Pr[Q\cap R \text{ covers }j] &= \Pr\left[j \in \bigcup_{i \in Q\cap R} A_i\right] = 1-\prod_{i: j \in A_i}\left(1-\frac{x_i}{p} \cdot p\right)\\ 
   &\geq 1-\exp\left( -\sum_{i: j \in A_i }x_i \right)\\
   &\geq 1-e^{y_i} \\
   &\geq \left(1-\frac{1}{e}\right) y_j. 
\end{align*}
Therefore, $$\Exp\left[f(Q \cap R)\right] \geq \left(1-\frac{1}{e}\right)\cdot\opt_{LP} \geq \left(1-\frac{1}{e}\right)\cdot\opt.$$ 
Observe that each set $i \in E$, appears in $Q \cap R$ independently with probability $x_i$. The set $Q \cap R$ may be infeasible in general (though it is feasible on average due to $x \in P_I$). We must lowerbound the objective value of the best feasible subset of $Q \cap R$, and we do this using monotone contention resolutions schemes. Let $\pi$ be a $c$-balanced monotone CRS instantiated for polytope $P_\I$. The (random) set $T=\pi_x(Q \intersect R)$ is feasible, and satisfies \[\Ex[f(T)] \geq c \Ex[f(Q \intersect R)] \geq c \cdot \left(1-\frac{1}{e}\right) \opt.\]  
We obtain the following theorem.


\begin{theorem}\label{thm:submodularsparsematroid}
Given a stochastic packing problem $\spack$ with coverage objective $f$ such that $P_\I$  admits a $c$-balanced monotone CRS, there exists $\left(1-\frac 1 e \right)\cdot c$-approximate sparsifier with degree $\frac 1 p$. Moreover, if $f$ is given explicitly and $P_I$ is polynomial-time solvable, then the sparsifier can be implemented in polynomial time.
\end{theorem}

Theorem~\ref{thm:submodularsparsematroid}, along with the monotone contention resolution schemes of \cite{alaei2014bayesian, chekuri2014submodular}, implies the following corollary.

\begin{corollary}
For stochastic packing problems with a coverage objective function, there is a sparsifier with degree $\frac{1}{p}$ achieving the following approximation ratios:
\begin{enumerate}
	\item $\left(1-\frac{1}{e}\right)^2$ for a matroid constraint.
	\item $\left(1 - \frac{1}{\sqrt {r+3}} \right)\left(1 - \frac{1}{e} \right)$ for an $r$-uniform matroid constraint.
    \item $\left(1 - \frac{1}{e} \right)^3$ for the intersection of two matroid constraints.
    \end{enumerate}
    Moroever, this sparsifier can be implemented in polynomial time when the coverage function objective is given explicitly, and the polytope of the constraints is polynomial-time solvable.
\end{corollary}

The above corollary implies that the approximation ratio approaches $\left(1-\frac{1}{e}\right)$ for an $r$-uniform matroid constraint as $r\rightarrow \infty$. In the next section, we show that this is the best possible approximation ratio of any sparsifier for a coverage objective and a uniform matroid constraint.




\subsection{Impossibility Result}
\label{sec:impossibility}
In this section, we show a certain limit on the sparsification of stochastic packing problems for submodular objectives even with an $r$-uniform matroid constraint. Formally, we show that for any fixed $\delta >0$, there is no $\left( 1-\frac{1}{e}+\delta \right)$-approximate sparsifier for a stochastic packing problem with a submodular objective and a matroid constraint with degree $O_p(1)$. Indeed, the impossibility result holds for coverage objectives and $r$-uniform matroid constraints, which is one of the most permissible problem instances of stochastic packing with submodular objective functions. Our impossibility result is information-theoretic and holds for any sparsifier with degree $O_p(1)$ \footnote{Here, $O_p(1)$ is a function only depending on $p$, and $O_p(1)$ is allowed to be much larger than $1/p$.} regardless of the runtime of the sparsifier.

\begin{theorem}
    \label{thm:impossible}
  For any for any $\delta >0$, there exists an SPP $\spack$ with a coverage objective, uniform matroid constraints and probability $p\in[0,1/3]$ such that no sparsifier with degree $O_p(1)$ achives approximation ratio better than $1-\frac 1 e + \delta $ .
\end{theorem}

Before proceeding with the proof, we first construct a family of instances $\spack$ with a coverage objective and $r$-uniform matroid constraint. We start by defining a coverage function on measurable subsets of $\mathbb{R}$. Let $f$ be a function such that $f(\mathcal{S}) = m\left( \bigcup_{S \in \mathcal{S}} S \right)$, where, $m(U)$ is the measure of the given set $U$. 


We call a collection of measurable sets $\mathcal{S}$ an equal $r$-partition of $[0,1]$ if elements $S_1, \dots S_k \subseteq [0,1]$ of $\mathcal{S}$ forms a partition of $[0,1]$ with an equal measure of $1/r$. Next, we describe the construction of $n/r$ many distinct equal $r$-partitions $\mathcal{S}^1, \mathcal{S}^2, \dots \mathcal{S}^{n/r}$ where $\S^i = \{ S_j^i\}_{j=1}^r$ and assume that $n$ is a multiple of $r$. We define  ground set $E=\bigcup_{i=1}^{n/r} \mathcal{S}^i$. Our aim is  to construct equal partitions in a way that any set $S^i_j \in \mathcal{S}^i$ covers $1/r$ fraction of any other set $S_{i'}^{j'}$ for $i' \neq i$ from a different $r$-partition. Such a construction allows to make all alternatives of $S^i_j$ symmetric. We obtain such a desired construction of $n/r$ many different equal $r$-partitions as follows: 
\begin{equation}\label{eq:equal-r-partition}
\mathcal{S}^i = \left\{\left. S^i_j := \bigcup_{\ell=1}^{r^{i-1}} \left[ \frac{\ell-1}{r^{i-1}}+\frac{j-1}{r^i}, \frac{\ell-1}{r^{i-1}}+\frac{j}{r^i} \right] \right| \forall j \in [r]\right\}. \qquad \text{for any } i \in [n/r].
\end{equation}
See Figure~\ref{fig:partitions} for the visualization of the construction. The following proposition demonstrates that any two partitions $\S^i$ and $\S^j$ are indistinguishable for the evaluation of function $f(\cdot)$.
\begin{proposition}
    \label{prop:sym}
    Let $Q \subseteq E = \bigcup_i \mathcal{S}^i$, and $\vec{s}$ be the \emph{incidence vector} such that $s_i := |\mathcal{S}^i \cap Q|$. Then,
    $$ f(Q) = 1 - \prod_{i=1}^{n/r} \left(1-\frac{s_i}{r}\right)  $$
\end{proposition}
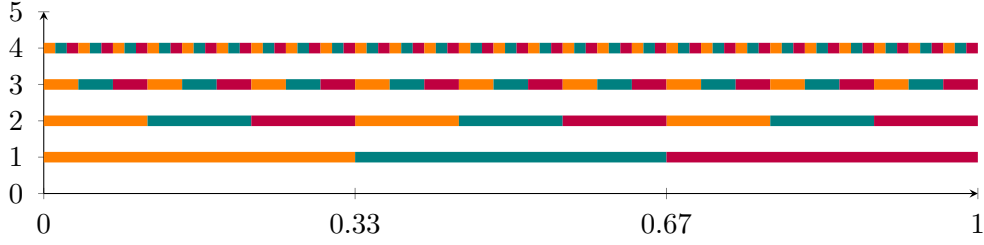
\begin{figure}[htb]
    \centering

\begin{tikzpicture}
    \begin{axis}[
        width=14 cm,
        height=4 cm,
        axis line style = thin,
        axis lines = left,
        xtick = {0,1/3,2/3,1},
        ytick = {0,1,2,3,4,5},
        ymin = 0,
        ymax = 5,
        line width=4,
    ]
    
    \addplot [
        domain=0:1/3, 
        samples=100, 
        color=orange,
    ]
    {+ 1};
    \addplot [
        domain=1/3:2/3, 
        samples=100, 
        color=teal,
    ]
    {+ 1};
    \addplot [
        domain=2/3:1, 
        samples=100, 
        color=purple,
    ]
    {+ 1};
    
    \addplot [
        domain=0:1/9, 
        samples=100, 
        color=orange,
    ]
    {+ 2};
    \addplot [
        domain=1/3:4/9, 
        samples=100, 
        color=orange,
    ]
    {+ 2};
    \addplot [
        domain=2/3:7/9, 
        samples=100, 
        color=orange,
    ]
    {+ 2};
    
    \addplot [
        domain=1/9:2/9, 
        samples=100, 
        color=teal,
    ]
    {+ 2};
    \addplot [
        domain=4/9:5/9, 
        samples=100, 
        color=teal,
    ]
    {+ 2};
    \addplot [
        domain=7/9:8/9, 
        samples=100, 
        color=teal,
    ]
    {+ 2};
    
    \addplot [
        domain=2/9:1/3, 
        samples=100, 
        color=purple,
    ]
    {+ 2};
    \addplot [
        domain=5/9:2/3, 
        samples=100, 
        color=purple,
    ]
    {+ 2};
    \addplot [
        domain=8/9:1, 
        samples=100, 
        color=purple,
    ]
    {+ 2};
    
    
    \addplot[
        domain=0/27:1/27, 
        samples=100, 
        color=orange,
    ] {+ 3};

    \addplot[
        domain=1/27:2/27, 
        samples=100, 
        color=teal,
    ] {+ 3};

    \addplot[
        domain=2/27:3/27, 
        samples=100, 
        color=purple,
    ] {+ 3};

    \addplot[
        domain=3/27:4/27, 
        samples=100, 
        color=orange,
    ] {+ 3};

    \addplot[
        domain=4/27:5/27, 
        samples=100, 
        color=teal,
    ] {+ 3};

    \addplot[
        domain=5/27:6/27, 
        samples=100, 
        color=purple,
    ] {+ 3};

    \addplot[
        domain=6/27:7/27, 
        samples=100, 
        color=orange,
    ] {+ 3};

    \addplot[
        domain=7/27:8/27, 
        samples=100, 
        color=teal,
    ] {+ 3};

    \addplot[
        domain=8/27:9/27, 
        samples=100, 
        color=purple,
    ] {+ 3};

    \addplot[
        domain=9/27:10/27, 
        samples=100, 
        color=orange,
    ] {+ 3};

    \addplot[
        domain=10/27:11/27, 
        samples=100, 
        color=teal,
    ] {+ 3};

    \addplot[
        domain=11/27:12/27, 
        samples=100, 
        color=purple,
    ] {+ 3};

    \addplot[
        domain=12/27:13/27, 
        samples=100, 
        color=orange,
    ] {+ 3};

    \addplot[
        domain=13/27:14/27, 
        samples=100, 
        color=teal,
    ] {+ 3};

    \addplot[
        domain=14/27:15/27, 
        samples=100, 
        color=purple,
    ] {+ 3};

    \addplot[
        domain=15/27:16/27, 
        samples=100, 
        color=orange,
    ] {+ 3};

    \addplot[
        domain=16/27:17/27, 
        samples=100, 
        color=teal,
    ] {+ 3};

    \addplot[
        domain=17/27:18/27, 
        samples=100, 
        color=purple,
    ] {+ 3};

    \addplot[
        domain=18/27:19/27, 
        samples=100, 
        color=orange,
    ] {+ 3};

    \addplot[
        domain=19/27:20/27, 
        samples=100, 
        color=teal,
    ] {+ 3};

    \addplot[
        domain=20/27:21/27, 
        samples=100, 
        color=purple,
    ] {+ 3};

    \addplot[
        domain=21/27:22/27, 
        samples=100, 
        color=orange,
    ] {+ 3};

    \addplot[
        domain=22/27:23/27, 
        samples=100, 
        color=teal,
    ] {+ 3};

    \addplot[
        domain=23/27:24/27, 
        samples=100, 
        color=purple,
    ] {+ 3};

    \addplot[
        domain=24/27:25/27, 
        samples=100, 
        color=orange,
    ] {+ 3};

    \addplot[
        domain=25/27:26/27, 
        samples=100, 
        color=teal,
    ] {+ 3};

    \addplot[
        domain=26/27:27/27, 
        samples=100, 
        color=purple,
    ] {+ 3};

    
    \addplot[
        domain=0/81:1/81, 
        samples=100, 
        color=orange,
    ] {+ 4};

    \addplot[
        domain=1/81:2/81, 
        samples=100, 
        color=teal,
    ] {+ 4};

    \addplot[
        domain=2/81:3/81, 
        samples=100, 
        color=purple,
    ] {+ 4};

    \addplot[
        domain=3/81:4/81, 
        samples=100, 
        color=orange,
    ] {+ 4};

    \addplot[
        domain=4/81:5/81, 
        samples=100, 
        color=teal,
    ] {+ 4};

    \addplot[
        domain=5/81:6/81, 
        samples=100, 
        color=purple,
    ] {+ 4};

    \addplot[
        domain=6/81:7/81, 
        samples=100, 
        color=orange,
    ] {+ 4};

    \addplot[
        domain=7/81:8/81, 
        samples=100, 
        color=teal,
    ] {+ 4};

    \addplot[
        domain=8/81:9/81, 
        samples=100, 
        color=purple,
    ] {+ 4};

    \addplot[
        domain=9/81:10/81, 
        samples=100, 
        color=orange,
    ] {+ 4};

    \addplot[
        domain=10/81:11/81, 
        samples=100, 
        color=teal,
    ] {+ 4};

    \addplot[
        domain=11/81:12/81, 
        samples=100, 
        color=purple,
    ] {+ 4};

    \addplot[
        domain=12/81:13/81, 
        samples=100, 
        color=orange,
    ] {+ 4};

    \addplot[
        domain=13/81:14/81, 
        samples=100, 
        color=teal,
    ] {+ 4};

    \addplot[
        domain=14/81:15/81, 
        samples=100, 
        color=purple,
    ] {+ 4};

    \addplot[
        domain=15/81:16/81, 
        samples=100, 
        color=orange,
    ] {+ 4};

    \addplot[
        domain=16/81:17/81, 
        samples=100, 
        color=teal,
    ] {+ 4};

    \addplot[
        domain=17/81:18/81, 
        samples=100, 
        color=purple,
    ] {+ 4};

    \addplot[
        domain=18/81:19/81, 
        samples=100, 
        color=orange,
    ] {+ 4};

    \addplot[
        domain=19/81:20/81, 
        samples=100, 
        color=teal,
    ] {+ 4};

    \addplot[
        domain=20/81:21/81, 
        samples=100, 
        color=purple,
    ] {+ 4};

    \addplot[
        domain=21/81:22/81, 
        samples=100, 
        color=orange,
    ] {+ 4};

    \addplot[
        domain=22/81:23/81, 
        samples=100, 
        color=teal,
    ] {+ 4};

    \addplot[
        domain=23/81:24/81, 
        samples=100, 
        color=purple,
    ] {+ 4};

    \addplot[
        domain=24/81:25/81, 
        samples=100, 
        color=orange,
    ] {+ 4};

    \addplot[
        domain=25/81:26/81, 
        samples=100, 
        color=teal,
    ] {+ 4};

    \addplot[
        domain=26/81:27/81, 
        samples=100, 
        color=purple,
    ] {+ 4};

    \addplot[
        domain=27/81:28/81, 
        samples=100, 
        color=orange,
    ] {+ 4};

    \addplot[
        domain=28/81:29/81, 
        samples=100, 
        color=teal,
    ] {+ 4};

    \addplot[
        domain=29/81:30/81, 
        samples=100, 
        color=purple,
    ] {+ 4};

    \addplot[
        domain=30/81:31/81, 
        samples=100, 
        color=orange,
    ] {+ 4};

    \addplot[
        domain=31/81:32/81, 
        samples=100, 
        color=teal,
    ] {+ 4};

    \addplot[
        domain=32/81:33/81, 
        samples=100, 
        color=purple,
    ] {+ 4};

    \addplot[
        domain=33/81:34/81, 
        samples=100, 
        color=orange,
    ] {+ 4};

    \addplot[
        domain=34/81:35/81, 
        samples=100, 
        color=teal,
    ] {+ 4};

    \addplot[
        domain=35/81:36/81, 
        samples=100, 
        color=purple,
    ] {+ 4};

    \addplot[
        domain=36/81:37/81, 
        samples=100, 
        color=orange,
    ] {+ 4};

    \addplot[
        domain=37/81:38/81, 
        samples=100, 
        color=teal,
    ] {+ 4};

    \addplot[
        domain=38/81:39/81, 
        samples=100, 
        color=purple,
    ] {+ 4};

    \addplot[
        domain=39/81:40/81, 
        samples=100, 
        color=orange,
    ] {+ 4};

    \addplot[
        domain=40/81:41/81, 
        samples=100, 
        color=teal,
    ] {+ 4};

    \addplot[
        domain=41/81:42/81, 
        samples=100, 
        color=purple,
    ] {+ 4};

    \addplot[
        domain=42/81:43/81, 
        samples=100, 
        color=orange,
    ] {+ 4};

    \addplot[
        domain=43/81:44/81, 
        samples=100, 
        color=teal,
    ] {+ 4};

    \addplot[
        domain=44/81:45/81, 
        samples=100, 
        color=purple,
    ] {+ 4};

    \addplot[
        domain=45/81:46/81, 
        samples=100, 
        color=orange,
    ] {+ 4};

    \addplot[
        domain=46/81:47/81, 
        samples=100, 
        color=teal,
    ] {+ 4};

    \addplot[
        domain=47/81:48/81, 
        samples=100, 
        color=purple,
    ] {+ 4};

    \addplot[
        domain=48/81:49/81, 
        samples=100, 
        color=orange,
    ] {+ 4};

    \addplot[
        domain=49/81:50/81, 
        samples=100, 
        color=teal,
    ] {+ 4};

    \addplot[
        domain=50/81:51/81, 
        samples=100, 
        color=purple,
    ] {+ 4};

    \addplot[
        domain=51/81:52/81, 
        samples=100, 
        color=orange,
    ] {+ 4};

    \addplot[
        domain=52/81:53/81, 
        samples=100, 
        color=teal,
    ] {+ 4};

    \addplot[
        domain=53/81:54/81, 
        samples=100, 
        color=purple,
    ] {+ 4};

    \addplot[
        domain=54/81:55/81, 
        samples=100, 
        color=orange,
    ] {+ 4};

    \addplot[
        domain=55/81:56/81, 
        samples=100, 
        color=teal,
    ] {+ 4};

    \addplot[
        domain=56/81:57/81, 
        samples=100, 
        color=purple,
    ] {+ 4};

    \addplot[
        domain=57/81:58/81, 
        samples=100, 
        color=orange,
    ] {+ 4};

    \addplot[
        domain=58/81:59/81, 
        samples=100, 
        color=teal,
    ] {+ 4};

    \addplot[
        domain=59/81:60/81, 
        samples=100, 
        color=purple,
    ] {+ 4};

    \addplot[
        domain=60/81:61/81, 
        samples=100, 
        color=orange,
    ] {+ 4};

    \addplot[
        domain=61/81:62/81, 
        samples=100, 
        color=teal,
    ] {+ 4};

    \addplot[
        domain=62/81:63/81, 
        samples=100, 
        color=purple,
    ] {+ 4};

    \addplot[
        domain=63/81:64/81, 
        samples=100, 
        color=orange,
    ] {+ 4};

    \addplot[
        domain=64/81:65/81, 
        samples=100, 
        color=teal,
    ] {+ 4};

    \addplot[
        domain=65/81:66/81, 
        samples=100, 
        color=purple,
    ] {+ 4};

    \addplot[
        domain=66/81:67/81, 
        samples=100, 
        color=orange,
    ] {+ 4};

    \addplot[
        domain=67/81:68/81, 
        samples=100, 
        color=teal,
    ] {+ 4};

    \addplot[
        domain=68/81:69/81, 
        samples=100, 
        color=purple,
    ] {+ 4};

    \addplot[
        domain=69/81:70/81, 
        samples=100, 
        color=orange,
    ] {+ 4};

    \addplot[
        domain=70/81:71/81, 
        samples=100, 
        color=teal,
    ] {+ 4};

    \addplot[
        domain=71/81:72/81, 
        samples=100, 
        color=purple,
    ] {+ 4};

    \addplot[
        domain=72/81:73/81, 
        samples=100, 
        color=orange,
    ] {+ 4};

    \addplot[
        domain=73/81:74/81, 
        samples=100, 
        color=teal,
    ] {+ 4};

    \addplot[
        domain=74/81:75/81, 
        samples=100, 
        color=purple,
    ] {+ 4};

    \addplot[
        domain=75/81:76/81, 
        samples=100, 
        color=orange,
    ] {+ 4};

    \addplot[
        domain=76/81:77/81, 
        samples=100, 
        color=teal,
    ] {+ 4};

    \addplot[
        domain=77/81:78/81, 
        samples=100, 
        color=purple,
    ] {+ 4};

    \addplot[
        domain=78/81:79/81, 
        samples=100, 
        color=orange,
    ] {+ 4};

    \addplot[
        domain=79/81:80/81, 
        samples=100, 
        color=teal,
    ] {+ 4};

    \addplot[
        domain=80/81:81/81, 
        samples=100, 
        color=purple,
    ] {+ 4};

    \end{axis}
\end{tikzpicture}

    \caption{Visualization of construction for $n=12$ and $r=3$ on $[0,1]$. Three colors at level $y=i$ for $i=1,\dots 4$ denotes $3$ measurable sets of equal $3$-partition.} \label{fig:partitions}
\end{figure}

\begin{proof}
We prove the proposition by induction. We define $U(Q)$ be the uncovered fraction of measure space by set $Q \subseteq \bigcup_{i}\mathcal S^i$. We use induction on the number of partitions. We claim that when the number of $r$-partition is $t$, uncovered fraction is $U(Q)=\prod_{i=1}^{t} \left(1-\frac{s_i}{r}\right)$. 
    
For $t=1$, the claim follows trivially. Now, let by the inductive hypothesis, the claim holds when the number of $r$-partitions is $\leq t$. We need to show that the claim holds when the number of $r$-partitions are $t+1$. 

Let $Q \subseteq \bigcup_{i=1}^{t+1} \mathcal{S}^i$. Fix any set $S_j^1 \in \mathcal{S}^1 \setminus Q$ which is not covered by $Q$. Now, we will use induction to compute what fraction of $S^1_j$ is covered by $Q \setminus \mathcal{S}^1$. Let us define measurable sets restricted to interval $S^1_j = \left[ \frac{j-1}{r}, \frac{j}{r} \right]$. Let $\mathcal{T}^i := \{S_j^1 \cap S  \mid S \in \mathcal{S}^i\}$ for all $i=2, \dots t+1$, and $\mathcal{T} := \bigcup_{i=2}^{t+1} \mathcal{T}^i$. Observe that each $\mathcal{T}^i$ is an equal $r$-partition of $S_j^1$, and $\mathcal{T}$ is the union of $t$ distinct equal $r$-partitions of $S^1_j$. Similarly, we can define the restriction of $Q$ to the interval $S^1_j$, that is $\bar{Q}:=\{S^1_j \cap S \mid S \in Q \setminus \mathcal{S}^1\}$. Intuitively, once we zoom in to interval $S^1_j$, $\mathcal{T}$ will look like $\mathcal{S} \setminus \mathcal{S}^{t+1}$ if we scale everything by $r$, see Figure~\ref{fig:partition2}. Now, we can apply induction on $\mathcal{T}$ and $\bar{Q}$ to compute what fraction of the interval $S_j^1$ is covered by $\bar{Q}$. 

To make the discussion formal, let us define a mapping $\rho: \left[\frac{j-1}{r}, \frac{j}{r}\right] \rightarrow [0,1]$ such that $\rho(x)=\left(x-\frac{j-1}{r} \right)\cdot r$, and extend it $\rho[T]=\{\rho(x) \mid x \in T\}$ to subsets $T \subseteq \left[\frac{j-1}{r}, \frac{j}{r}\right]$. Observe that $\{\rho[T] \mid T \in \mathcal{T}\}=\mathcal{S} \setminus \mathcal{S}^{t+1}$ and $\rho$ multiplies measure of the sets in $\mathcal{T}$ by $r$. Therefore, if we apply induction to $\mathcal{T}$ with $\bar{Q}$, we obtain that 
$$ m(\bar{Q})=\frac{1}{r}\cdot\left(1- \prod_{i=2}^{t+1} \left( 1- \frac{s_i}{r} \right)\right).$$
Recall that $m(Q)$ is the measure of the union of sets in $Q$. The equality follows from the fact that $m(S^1_j)=1/r$ and $S^1_j$ intersects with any $S \in Q\setminus\mathcal{S}^1$. Notice that the above equality holds for any $S^1_j$. Therefore, we obtain
$$ f(Q) = m(Q) = 1-\prod_{i=1}^{t+1} \left(1-\frac{s_i}{r} \right) $$
concluding the proof.

\end{proof}

\begin{figure}[htb]
    \centering

\begin{tikzpicture}
    \begin{axis}[
        width=14 cm,
        height=4 cm,
        axis line style = thin,
        axis lines = left,
        xtick = {0,1/9,2/9,3/9},
        ytick = {0,1,2,3,4,5},
        ymin = 0,
        ymax = 5,
        line width=4,
    ]
    \addplot [
        domain=0:1/3, 
        samples=100, 
        color=orange,
    ]
    {+ 1};

    \addplot [
        domain=0:1/9, 
        samples=100, 
        color=orange,
    ]
    {+ 2};
    \addplot [
        domain=1/9:2/9, 
        samples=100, 
        color=teal,
    ]
    {+ 2};
    \addplot [
        domain=2/9:3/9, 
        samples=100, 
        color=purple,
    ]
    {+ 2};
    
    \addplot [
        domain=0:1/27, 
        samples=100, 
        color=orange,
    ]
    {+ 3};
    \addplot [
        domain=1/9:4/27, 
        samples=100, 
        color=orange,
    ]
    {+ 3};
    \addplot [
        domain=2/9:7/27, 
        samples=100, 
        color=orange,
    ]
    {+ 3};
    
    \addplot [
        domain=1/27:2/27, 
        samples=100, 
        color=teal,
    ]
    {+ 3};
    \addplot [
        domain=4/27:5/27, 
        samples=100, 
        color=teal,
    ]
    {+ 3};
    \addplot [
        domain=7/27:8/27, 
        samples=100, 
        color=teal,
    ]
    {+ 3};
    
    \addplot [
        domain=2/27:1/9, 
        samples=100, 
        color=purple,
    ]
    {+ 3};
    \addplot [
        domain=5/27:2/9, 
        samples=100, 
        color=purple,
    ]
    {+ 3};
    \addplot [
        domain=8/27:1/3, 
        samples=100, 
        color=purple,
    ]
    {+ 3};
    
    
    \addplot[
        domain=0/81:1/81, 
        samples=100, 
        color=orange,
    ] {+ 4};

    \addplot[
        domain=1/81:2/81, 
        samples=100, 
        color=teal,
    ] {+ 4};

    \addplot[
        domain=2/81:3/81, 
        samples=100, 
        color=purple,
    ] {+ 4};

    \addplot[
        domain=3/81:4/81, 
        samples=100, 
        color=orange,
    ] {+ 4};

    \addplot[
        domain=4/81:5/81, 
        samples=100, 
        color=teal,
    ] {+ 4};

    \addplot[
        domain=5/81:6/81, 
        samples=100, 
        color=purple,
    ] {+ 4};

    \addplot[
        domain=6/81:7/81, 
        samples=100, 
        color=orange,
    ] {+ 4};

    \addplot[
        domain=7/81:8/81, 
        samples=100, 
        color=teal,
    ] {+ 4};

    \addplot[
        domain=8/81:9/81, 
        samples=100, 
        color=purple,
    ] {+ 4};

    \addplot[
        domain=9/81:10/81, 
        samples=100, 
        color=orange,
    ] {+ 4};

    \addplot[
        domain=10/81:11/81, 
        samples=100, 
        color=teal,
    ] {+ 4};

    \addplot[
        domain=11/81:12/81, 
        samples=100, 
        color=purple,
    ] {+ 4};

    \addplot[
        domain=12/81:13/81, 
        samples=100, 
        color=orange,
    ] {+ 4};

    \addplot[
        domain=13/81:14/81, 
        samples=100, 
        color=teal,
    ] {+ 4};

    \addplot[
        domain=14/81:15/81, 
        samples=100, 
        color=purple,
    ] {+ 4};

    \addplot[
        domain=15/81:16/81, 
        samples=100, 
        color=orange,
    ] {+ 4};

    \addplot[
        domain=16/81:17/81, 
        samples=100, 
        color=teal,
    ] {+ 4};

    \addplot[
        domain=17/81:18/81, 
        samples=100, 
        color=purple,
    ] {+ 4};

    \addplot[
        domain=18/81:19/81, 
        samples=100, 
        color=orange,
    ] {+ 4};

    \addplot[
        domain=19/81:20/81, 
        samples=100, 
        color=teal,
    ] {+ 4};

    \addplot[
        domain=20/81:21/81, 
        samples=100, 
        color=purple,
    ] {+ 4};

    \addplot[
        domain=21/81:22/81, 
        samples=100, 
        color=orange,
    ] {+ 4};

    \addplot[
        domain=22/81:23/81, 
        samples=100, 
        color=teal,
    ] {+ 4};

    \addplot[
        domain=23/81:24/81, 
        samples=100, 
        color=purple,
    ] {+ 4};

    \addplot[
        domain=24/81:25/81, 
        samples=100, 
        color=orange,
    ] {+ 4};

    \addplot[
        domain=25/81:26/81, 
        samples=100, 
        color=teal,
    ] {+ 4};

    \addplot[
        domain=26/81:27/81, 
        samples=100, 
        color=purple,
    ] {+ 4};
    
    \end{axis}
\end{tikzpicture}

    \caption{Visualization of measurable sets restricted to interval $S^1_1=[0,1/3]$ when $n=12$ and $r=3$.}\label{fig:partition2}
\end{figure}
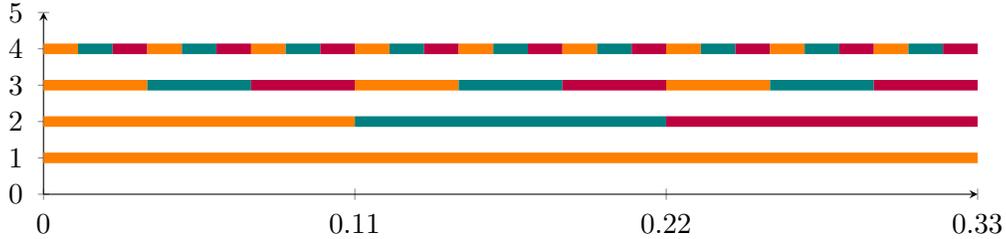

For the sake of exposition, we define
\begin{equation*}
    f(\vec{s}) := 1 - \prod_{i=1}^{n/r} \left(1-\frac{s_i}{r}\right).
\end{equation*}

In the following proposition, we prove how $f(\cdot)$ improves by local steps.
\begin{proposition}
    \label{prop:local_improve}
    Let $\vec{s} \in [r]^{n/r}$ with $s_i < r$ for all $i \in [n/r]$ be an incidence vector such that $r > s_i > s_j >0$ for some $i \neq j$. Define the new vector $\vec{s}'$ as $s'_i = s_i+1$, $s'_j = s_j-1$, and $s'_\ell = s$ for any $\ell \notin \{i,j\}$. Then, we have 
    $$ f(\vec{s}) < f(\vec{s}'). $$
\end{proposition}
\begin{proof}
    Let $\beta=\prod_{\ell \in [n/r]\setminus \{i,j\}} (1-\frac{s_\ell}{r})$. Then we have
    $$f(\vec{s}')-f(\vec{s}) = \beta \cdot \left[\left(1-\frac{s_i}{r}\right)\left(1-\frac{s_j}{r}\right)-\left(1-\frac{s_i}{r}-\frac{1}{r}\right)\left(1-\frac{s_j}{r}+\frac{1}{r}\right)\right] = \beta \cdot \left[\frac{s_i}{r^2} - \frac{s_j}{r^2}\right] > 0$$
\end{proof}
The proposition indicates that a sparse vector $\vec s$ implies higher value of the function $f$. As a corollary of Proposition~\ref{prop:local_improve}, $\max_{|Q|<r}f(Q)$ is maximized for $Q$ that selects sets from the same partition $\S^i$ for some $i\in [n/r]$. Next,we prove the desired upperbound on the approximation ratio of any degree $O_p(1)$ sparsifier for SPP $\spack$ and complete the proof of Theorem~\ref{thm:impossible}.

\begin{proof}[Proof of Theorem~\ref{thm:impossible}] Let $Q$ be a sparsifier of SSP $\spack$ with degree $O_p(1)$. WLOG, assume that $n$ is a multiple of $r$. Let $\mathcal{Q} \subseteq [n/r]$ be the set of indices of equal $r$-partitions (defined in Equation~\ref{eq:equal-r-partition}) with non-empty intersection with set $Q$.   
	
 We say that equal $r$-partition $\S^i$ is \emph{excessively active} if $\S^i \cap R $ covers at least $(1+\delta)\cdot p$ fraction of the measure space $[0,1]$. For any $i\in \Q$, usinf Chernoff's bound, we can upper bound the probability of any $r$-partition $\S^i$ being excessively active as follows:   
     \begin{align*}
 	\Pr[\text{$\mathcal{S}^i$ is excessively active}]
 	&=\Pr[|\mathcal{S}^i \cap R| > (1+\delta) \cdot r \cdot p]\\
 	&\leq \exp\left( - \frac{\delta^2 \cdot r \cdot p}{3}\right) && \text{(Chernoff's Bound)}
 \end{align*}
  By union bound, we can bound the probability of at least one of the equal $r$- partitions whose index is in $\Q$ covering $(1+\delta)\cdot p$ fraction of the measure space as follows:
\begin{align*}
	\Pr[\text{at least one $\mathcal{S}^i$ is excessively active for } i\in \Q ] \leq \frac{r}{O_p(1)} \cdot \exp \left( -\frac{\delta^2 \cdot r \cdot p}{3} \right) \leq \delta.
\end{align*}
Note that $|\mathcal{Q}| \leq |Q| \leq r/{O_p(1)}$ and the last inequality follows by plugging any $r > \log \left( {\delta \cdot O_p(1)} \right) \cdot \frac{3\cdot O_p(1)}{\delta^2}$. Let $\mathcal{E}$ be the event where none of the sets $\mathcal{S}^i$ for $i \in \mathcal{Q}$ is excessively active, i.e. $\forall i\in \mathcal Q: |\mathcal{S}^i \cap R| \leq (1+\delta) \cdot r \cdot p$. Note that $\Pr[\mathcal{E}] \geq 1 - \delta $.


 Let $\optsol(A)$ denotes the optimal solution in the set $A$, i.e. $\max_{A^* \subseteq A: |A^*|\leq r}f(A^*)$. Next, we first condition on the event $\mathcal E$ and upper bound $f(Q\cap R)$ as,  
\begin{align*}
	f(\optsol(Q\cap R))&=f\left( \optsol\left(R \cap \bigcup_{i \in \Q} \bigcup_{S \in \mathcal{S}^i} S\right) \right)\\
	 & \leq 1 - \prod_{i \in \Q} \left( 1 - \frac{|\S_i \cap R|}{r}\right) && \text{(by Propositions~\ref{prop:sym})} \\
	&\leq 1-\left[1-(1+\delta)\cdot p\right]^{1/p} && \text{(by Propositions~\ref{prop:local_improve} and event $\mathcal E$)}\\
	&\leq 1-e^{-(1+\delta)} && (\text{since $p \leq 1/3$}).
\end{align*}
Now, we can upper bound the expected value of $\optsol(Q\cap R)$: 
\begin{align*}
	\Ex[f(\optsol(Q) \cap R))] &= 
	\Ex[f(\optsol(Q) \cap R))\mid \mathcal{E}] \cdot \Pr[\mathcal{E}] + \Ex[f(\optsol((Q) \cap R))\mid \mathcal{E}^C] \cdot \Pr[\mathcal{E}^C] \\
	&\leq 1-e^{-(1+\delta)} + \delta \\
	&= 1-e^{-1}+e^{-1}(1-e^{-\delta})\\
	&\leq 1-e^{-1}+e^{-1}\delta \qquad \qquad \qquad (1-\delta \leq e^{-\delta})\\
	&\leq 1-\frac{1}{e}+2\delta.
\end{align*}

On the other hand, if $n/r > p^r \cdot 3 \log(1/\delta)$ then at least one of the equal $r$-partitions, say $\mathcal{S}^i$, will be completely active, with probability $(1-\delta)$, i.e. $\Pr[\exists i\in [n/r]: \mathcal{S}^i \subseteq R]\geq 1-\delta $. Formally, 
$$ \Pr[\text{at least one of $\mathcal{S}^i$ is completely active}] \geq 1-(1-p^r)^{n/r} \geq 1-\left(1-\frac{1}{e}\right)^{3 \log(1/\delta)} \geq 1-\delta. $$
Therefore, $\opt \geq 1-\delta$. Combining two facts proves that 
$$ \frac{\Ex[f(\optsol(Q \cap R))]}{\opt} \leq \frac{1-\frac{1}{e}+2\delta}{1-\delta} \leq 1-\frac{1}{e} + 3\delta $$
since $1-1/e \leq 1$. Replacing $\delta$ with $\delta'/3$ completes the proof.
\end{proof}

\section{Open Questions}
\begin{itemize}
	\item We believe that our results portend a deeper connection between the sparsification and contention resolution. The results of Section \ref{sec:CRS_sparse} show that contention resolution serves to lower-bound the sparsification ratio. We ask whether the connection goes both ways. In particular, does the existence of a $c$-sparsifier of degree $1/p$ imply a contention resolution scheme with balance $c$? This is intimated by Proposition~\ref{prop:both-ways}. Does the existence of a $c$-sparsifier of degree $\poly(1/p)$ imply a contention resolution scheme with balance $\Omega(c)$ (or some other expression involving $c$ and the degree)? This is intimated by the example from Appendix~\ref{sec:ssp_additive_impos}. Formalizing  a tighter connection between sparsification and contention resolution (equivalently, the correlation gap) might lead to new structural and computational insights for the latter.
	\item In Section~\ref{sec:sparse_matroid}, we show that a greedy sparsifier $1-\epsilon$ approximate with degree $O(1/p)$ for additive optimization subject to a matroid constraint. We conjecture that a similar greedy sparsifier exists for the intersection of $k$ matroids, obtaining a $\frac{1-\epsilon}{k}$-approximation with degree $O(1/p)$. A similar greedy sparsifier, albeit with degree $O(1/p^{1/\epsilon})$, was shown to be $1/2$-approximate for the special case of unweighted bipartite matching in \cite{blum2015ignorance}.
	\item We conjecture that there exists a $\frac{1-\eps}{k-1}$-approximate sparsifier with $\poly(1/p)$ degree for additive optimization subject to the intersection of $k$ matroid constraints. Our intuition comes from the work of \cite{lee2010submodular-k-intersect}, which shows a $\frac{1}{k-1}$ approximation algorithm for the associated optimization problem. The work of \cite{behnezhad2020unweighted,behnezhad2020unweighted} has already made some progress towards this conjecture by showing a $(1-\epsilon)$-approximate sparsifier for bipartite matching with degree $\exp(\exp(\exp(1/p)))$.

    \item Our results in Section~\ref{sec:improved_matching} improve the state of the art sparsifier for weighted (non-bipartite) matching in the polynomial degree regime. Moreover, since our approximation guarantee is a function of the correlation gap, progress on the correlation gap of the matching polytope will lead to further improved sparsifiers. Finding the best possible sparsification ratio in the polynomial degree regime remains open, however, with $1-\epsilon$ still on the table. Beyond polynomial degree, a $1-\epsilon$ approximate sparsifier with degree $\exp(\exp(\exp(1/p)))$ was already shown by \cite{behnezhad2020weighted}.

\end{itemize}



\bibliographystyle{alpha} 
\bibliography{ref,stochastic}

\appendix

\newpage
\section{Impossibility for  Sparsification of SPP with Additive Objectives}\label{sec:ssp_additive_impos}
\begin{theorem}
    There exists a SPP $\spack$ with an additive objective and downward closed packing constraints such that no sparsifier with degree $\frac 1 {p^c}$ for any constant $c>0$ is $\Omega(1)$ approximate.
\end{theorem}
\begin{proof}
We consider the following SPP with an unweighted additive objective function. Let $E = \bigcup_{i=1}^m E_i$ be the set of elements where $E_i = \{e_1^i,\dots e_k^i\}$ and $m = k^k \cdot \log k$. Therefore, $k = \omega(1)$. We define a downward closed constraint as follows: 
\begin{equation}\label{eq:constraints}
    \I = \left\{ I\subseteq E: \exists i\in [m] \text{ s.t. } I\subseteq E_i \right\},
\end{equation}
where the rank of the set system $\I$ is $k$ and the set of maximal feasible sets (also feasible sets with maximum size) is $\{E_1,\dots ,E_m\}$. We let $p = \frac 1 k$.

We now upper bound the approximation ratio of any $1/p^c$-degree sparsifier. First, 
let $Q$ be a sparsifier of SSP $\spack$ (possibly random) with degree is $1/p^c$. Let $F:=\{E_i : E_i \cap Q \neq \emptyset\}$ be the collection of subsets of $E_1,\dots, E_m$ such that $Q\cap E_i$ is non-empty. Note that $|F| \leq \frac{k}{p^c}$ as the degree of $Q$ is at most $\frac 1 {p^c}$ Observe that size of the optimal solution in $Q \cap R$ is upper-bounded by the maximum number of active elements in any $E_i \in F$. First we observe that $\E[|R\cap E_i|]\leq 1$. Therefore, using Chernoff bound and union bound, we claim that number of active elements in all $E_i \in F$ will be small with high probability. For any sparsifier $Q$ and the set $F$, Chernoff bound followed by union bound, we obtain
\begin{align*}
    \Pr\left[ \bigwedge_{E_i \in F} \left\{|E_i \cap R| \leq 1+\delta \right\} \right] & \leq 1-\left(\frac{k}{p^c}\right) \cdot e^{-\frac{\delta^2}{2+\delta}} \\
        &\leq 1- k^{c+1} \cdot e^{-\delta/2} && (p=1/k)\\
        &\leq 1-1/k && (\delta = (2c + 4)\log k).
\end{align*}
Therefore, $\Ex[|\optsol(Q \cap R)|] \leq \Ex[\max_{F_i \in F} |F_i \cap R|] \leq 1+(2c+4) \log k$. Since, $\log m \geq k$, we know that there exists a set $E_i$ such that $|E_i \cap R|=k$ with probability $(1-1/k)$. This implies that any $1/p^c$-degree sparsifier for some constant $c>0$ has an approximation ratio less than $\frac{O(\log (k))}{k}$.
\end{proof}

The results of Section~\ref{sec:CRS_sparse} imply that the correlation gap of this set system is also $O(\log k / k)$.

\section{Missing Proofs from the Main Paper}
\label{sec:missing_proofs}
\subsection{Proof of Theorem~\ref{thm:sparse_CRS}}
\crsestimation*
\begin{proof}
	Given a $\beta$-approximate stochastic optimum oracle $\asopt$, let $\tilde q_e = \Pr_{O\sim \asopt}[e\in O]$.  When $f$ is  additive with  weights $\bvec{w}$ we have $\sum_{e\in E}w_e\cdot \tilde q_e   \geq \beta \cdot \opt$. More generally, when $f$ is monotone submodular we have $f^+(\vec{\tilde q})\geq \E_{O \sim \asopt} \left[f(O)\right] \geq \beta \cdot \opt$.
	
	Using $\poly(n,\frac{1}{\epsilon})$ samples as well as Chernoff and union bounds, we obtain probabilistic understimates $\bvec{\hat{q}}$ of $\bvec{\tilde{q}}$ within an additive  $\frac \delta n$, where $\delta = \beta \epsilon p/2$. Formally, we guarantee that 
	$$ \Pr\left[ \forall e \in E : \tilde q_e - \frac{\delta}{n} \leq \hat q_e \leq \tilde q_e  \right] \geq 1-\epsilon/2.$$
	Going forward, we will analyze our sparsifier in the event that $\tilde q_e - \frac{\delta}{n} \leq \hat q_e \leq \tilde q_e $ for all $e \in E$.  This guarantees that our sparse set $Q$, constructed using $\vec{\hat{q}}$ in lieu of $\vec{q}$, is well-defined and has expected size at most $\frac{r}{p}$, as needed for our degree bound. It also guarantees that $\bvec{\hat{q}} \in \P_I$, so we can employ the same contention resolution arguments as in Theorem~\ref{thm:sparse_CRS_exact}. Otherwise, we allow our sparsifier to fail by outputting an arbitrary set of the expected sparsity.

	
	For an additive $f$, we have
	$$\sum_{e\in E} w_e\cdot \hat q_e\geq  \sum_{e\in E}w_e\cdot \tilde q_e - \delta \cdot \max_{e\in E} w_e \geq \beta \cdot \left(1-\epsilon/2\right)\opt,$$ where the last inequality follows from $p \cdot \max_{e \in E} w_e \leq \opt$. The proof then proceeds as in Theorem~\ref{thm:sparse_CRS_exact} to obtain
	$$ \E\left[ \max_{\substack{T\subseteq Q\cap R\\ T\in \I}} f(T) \right] \geq c \cdot \sum_{e \in E} w_e \cdot \hat q_e \geq c \cdot \beta \cdot (1-\epsilon/2) \cdot \opt. $$
	When $f$ is monotone submodular we have
	$$f^+(\vec{\hat q})\geq f^+(\vec{\tilde q}) - \delta \cdot \max_{e \in E} f(e) \geq \beta \cdot \left(1 - \epsilon/2 \right)\cdot \opt,$$
	where the last inequality follows from $p \cdot \max_{e \in E} f(e) \leq \opt.$
	As in Theorem~\ref{thm:sparse_CRS_exact}, we then obtain 
	$$ \E\left[ \max_{\substack{T\subseteq Q\cap R\\ T\in \I}} f(T) \right] \geq c \cdot \left(1-\frac 1 e\right) f^+(\vec{\hat{q}}) \geq c \cdot \beta \cdot (1-\epsilon/2) \cdot \left(1-\frac 1 e\right) \cdot \opt$$.

	We lose an additional factor of $1-\epsilon/2$ in our approximation ratios due to our probabilistic guarantee on the accuracy of $\vec{\hat{q}}$. This yields the claimed bounds.
\end{proof}

\subsection{Proof of Lemma~\ref{lem:small_prob}}

\begin{proof}
    Let $I = \match$ be an arbitrary instance of a stochastic weighted matching problem with graph $G:=(V,E)$ and a random set of active edges $R$ such that $\Pr[e \in R]=p$ for all $e \in E$. We construct an instance $J=\langle \tilde G, \tilde p \rangle$ of the same problem as follows: let $\tilde G = (V, \tilde E)$ be a graph with edge set $\tilde E = \bigcup_{e\in E} \{e^1,\dots,e^c\}$ where $c := \frac{\log \frac{1}{1-p}}{\log \frac{1}{1-\epsilon^4\cdot p}}$ and $\tilde f(\{e^j\}) = f(\{e\})$ for all $j\in [c]$. The set of active edges $\tilde R$ contains each edge $e^j \in \tilde E$ independently with probability $\tilde p:=\epsilon^4 \cdot p$, i.e. $\Pr[e^j \in \tilde R] = \tilde p$ for all $e \in E$ and $j \in [c]$. Observe that the graph $\tilde G$ contains $c$ many copies of each edge in the graph $G$ where $c$ is chosen in a way that for any edge $e\in E$, at least one copy of $e$ is active in $\tilde G$ with probability $p-o(\epsilon)$. More formally, for any $e\in E$,
    \begin{equation*}
        \Pr[\exists j\in [c]: e^j \in \tilde R] = 1- (1-\epsilon^4 \cdot p)^c = 1 - (1-p+o(\eps)) = p-o(\eps).
    \end{equation*}

    First of all, observe that the problem instance $J$ is a well-defined stochastic weighted matching problem where each edge is active with probability $\tilde p$ independently. In addition, WLOG, we can assume that $c$ is an integer since there is a way to pick a smaller $\epsilon$ and guarantee that $c$ is an integer. In the rest of the proof, we show by a coupling that any $\alpha$-approximate degree $d/\tilde p$ sparsifier for $J$ implies $\alpha$-approximate degree $d/\tilde p = d/(\epsilon^4 \cdot p)$ sparsifier for $I$ as well.
 
    Consider the coupling $\gamma$ between distributions of the set of active edges $(R, \tilde R)$ in $G$ and $\tilde G$ as follows.
\begin{itemize}
	\item Sample set of edges $\tilde R \subseteq \tilde E$ in the graph $\tilde G$ such that each edge $\tilde e\in \tilde E$ is sampled in the set $\tilde R$ independently with probability  $\epsilon^4 \cdot p$.
	\item For all $e\in E$, if any copy of edge $e$ is in $\tilde R$, i.e. $\exists j \in [c]: e^j \in \tilde R$, then add $e\in R$.
\end{itemize}
Since the set of copies is disjoint for two distinct edges $e, f \in E$, $R$ contains each edge $e \in E$ independently with probability $p$. Next, we show that the weight of the optimal weighted matching among the set of edges $R$ in a graph $G$ and the set of edges $\tilde R$ in graph $\tilde G$ are the same. Define $\optsol_I(R)$ as the maximum weight matching among the set of edges $R$ for the stochastic matching instance $I$, and similarly define $\optsol_J(\tilde R)$. Notice that for any $e\in \optsol_I(R)$, coupling $\gamma$ ensures that at least one of the corresponding copies of an edge $e$, say $e^j$ appears in $\tilde R$. Consider the set of edges $\tilde M$ in $\tilde G$ that selects exactly one copy of an edge $e \in \optsol_I(R)$ from $\tilde R$. Trivially, $\tilde M$ is a  subset of $\tilde R$ and forms a matching in $G$. Therefore, $\tilde f(\optsol_J(\tilde R))\geq \tilde f(\tilde M)=f(\optsol_I(R))$. Similarly, we show that $f(\optsol_I(R))\geq \tilde f(\optsol_J(\tilde R))$ which concludes that,
    \begin{equation*}
       \E_{(R,\tilde R)\sim \gamma}[ f(\optsol(R))] = \E_{(R,\tilde R)\sim \gamma}[\tilde f(\optsol(\tilde R))]
    \end{equation*}

Let $\tilde Q\subseteq E$ (possibly random) set of edges selected $\alpha$-approximate degree $d / \tilde{p}$ sparsifier of the problem $J$ which exists by the assumption of the lemma. Now, we construct a sparsifier $Q\subset E$ as follows: an edge $e\in E$ is added to the set $Q$ if at least one copy of the edge $e$ is selected in $\tilde Q$. Note that the degree of the sparsifier $Q$ is at most $d/\tilde p$. 
Let $\tilde M$ be the optimal solution of $\tilde R \cap  \tilde Q$, i.e. $\tilde M := \optsol_J(\tilde{R} \cap \tilde{Q})$, we construct a matching $M$ on the set of edges $Q\cap R$ as follows: add an edge $e \in Q\cap R$ to the matching $M$ if any copy of $e$ appears in $\tilde M$. Therefore,
$$ \E_R[f(\optsol_I(R \cap Q))] \geq \E[f(M)] = \E[\tilde f(\tilde M)] = \E_{\tilde R}[\tilde f(\optsol_J(\tilde R \cap \tilde Q))] \geq \alpha \cdot \E_{\tilde R}[\tilde f(\optsol_J(\tilde R))]$$
concludes the proof.
concludes the proof.
\end{proof}

\subsection{Proof  of Lemma~\ref{lem:non-crucial}}

Before we prove the lemma, we recall the procedure for constructing fractional matchig on $Q_\greedy \cap \nc \cap R$ \cite{behnezhad2019stochastic}. At last, we construct an intergal matching $M_\nc$ on  $Q_\greedy \cap \nc \cap R$  via Caratheodory's decomposition. 
    \begin{algorithm}[H]
        \caption{Construction of a random matching $M$ on non-crucial active edges of $Q_{\greedy}$ \label{alg:non-crucial-matching}}
        \begin{algorithmic}[1]
        \State Set $x_e \gets 0$ for all $e \in Q_{\greedy} \cap R$.
        \For {all edges $e \in Q_{\greedy} \cap \nc \cap R$}
            \State Compute $f_e := \frac 1 T \sum_{i=1}^T \mathbbm{1}[e \in Q_i]$ where $Q_i$ is the sample obtained in the $i^{\text{th}}$ round of Algorithm~\ref{alg:matching_sparsify}.
            \State Set $\tilde{x}_e \gets \min \left\{ f_e/p, 2\tau/p \right\}$
        \EndFor
        \State For each edge $e = (u,v)$ define the \emph{scaling factor} $s_e$ as 
        $$ s_e := \min \left\{ 1, \frac{\max\{q^\nc_v, \epsilon\}}{\tilde{x}_v},\frac{\max\{q^\nc_u, \epsilon\}}{\tilde{x}_u}  \right\}$$
        where $q^\nc_v:=\sum_{e: e\in \nc, v \in e} q_e$ and $\tilde x_v:= \sum_{e: e\in \nc, v \in e} \tilde  x_e$.
        \State Scale down the fractional matching $x_e \gets \tilde{x_e} \cdot s_e$.
        \State Sample a random matching $M_\nc$ such that $\Pr[e \in M_\nc] = x_e$ for all $e \in Q_{\greedy} \cap \nc \cap R$ by Caratheodory's decompositi on.
        \State \textbf{Output}: $M_\nc$.
        \end{algorithmic}
    \end{algorithm}
We first recall some  important properties of the fractional matching constructed by the above procedure.  First claim shows that for any non-crucial edge $e$, $\tilde x_e$ is approximately $q_e$ in expectation. 
\begin{claim}{\cite{behnezhad2019stochastic}[Claim A.1., Claim A.2]}\label{claim:a2}
For any non-crucial edge $e$, $\Pr\left[\tilde x_e = \frac{2\tau}{p}   \right]\leq \eps q_e$. Moreover, 
\begin{equation*}
	\E[\tilde x_e] = \E\left[\frac 1 p\min\{f_e ,2\tau\} \right] \geq (1-\eps) \cdot \frac{q_e}{p} \quad \text{  and  } \quad 	\E[\tilde x_e\mid \tilde x_e  < 2\tau/p] \geq (1 - 3\cdot \eps )\cdot \frac{q_e}{p} 
\end{equation*}
\end{claim}
For any non-crucial edge $e $ and vertex $v\in e$, let $e_1,\dots, e_k$ be the non-crucial edges in $Q_\greedy\cap R$ incident to $v$ except $e$. Let $f_{v\setminus e}^\nc = \sum_{e_i} f_{e_i}$. The following claim shows that $f_{v\setminus e}^\nc $ approximates $\sum_{e_i} q_{e_i}$. 
\begin{claim}{\cite{behnezhad2019stochastic}[Claim A.5].} \label{claim:a5}
	For any non-crucial edge $e$, 
	\begin{equation*}
		 \Pr[\max\{f_{v\setminus e}^\nc ,\eps\} \leq (1+\eps) \max \{ q_v^\nc ,\eps \}]\geq 1-\eps
	\end{equation*}
\end{claim}
For any non-crucial edge $e=(u,v) \in \nc$, we let $\tilde x_{v\setminus e} = \sum_{e': e'\in \nc\setminus e, v \in e'} \tilde  x_{'e}$ and $\tilde x_{u\setminus e} = \sum_{e': e'\in \nc\setminus e, u \in e'} \tilde  x_{'e}$.  For each $e\in \nc$, we define $s_e'$ as follows: 
\begin{equation}\label{eq:worst-scaling-factor}
	s'_e := \min \left\{ 1, \frac{\max\{q^\nc_v , \epsilon\}}{\tilde{x}_{v\setminus e} +\frac{2\tau}{p} },\frac{\max\{q^\nc_u, \epsilon\}}{\tilde{x}_{u\setminus e}+\frac{2\tau}{p} }  \right\}
\end{equation}
First, we observe $\tilde x_e \leq \frac{2\tau}{p}$, therefore, $s'_e \leq s_e$ with probability $1$. In next claim, we show that $s_e'$ close to $1$ with high probability.
\begin{claim}\label{claim:s-eis'1}
For any non-crucial edge $e\in \nc$,
\begin{align*}
	\Pr[s_e' \geq 1-2\eps] \geq 1-2\eps.
\end{align*}
\end{claim}
\begin{proof}
Claim~\ref{claim:a5} implies that 
\begin{align*}
 \Pr[\max \{\tilde x_{v\setminus e}  +  2\tau /p, \eps \} \leq (1+\eps)  \max \{ q_v^\nc ,\eps \} ]
	&\geq \Pr[\max \{\tilde x_{v\setminus e}  +  2\tau /p, \eps \} \leq (1+\eps)  \max \{ q_v^\nc ,\eps \} - \eps^2 ]\\
	&\geq \Pr[\max \{\tilde x_{v\setminus e}  +  2\tau /p, \eps \} \leq (1+2\eps)  \max \{ q_v^\nc ,\eps \} ]\\
	&\geq 1-\eps .
\end{align*}
Where the first inequailty holds becasue $\tau/p \leq \eps^2$ and the last inequality holds from Lemma~\ref{claim:a5}. By union-bound, we obtian 
\begin{equation*}
	\Pr\left[\min \left\{ 1, \frac{\max\{q^\nc_v, \epsilon\}}{\tilde{x}_{v\setminus e} + 2\tau/p},\frac{\max\{q^\nc_u, \epsilon\}}{\tilde{x}_{u\setminus e}+2\tau/p} \right\} \geq \frac{1}{1+2\eps} \right] =  \Pr \left[ s'_e \geq  \frac 1 {1+2\eps}\right] \geq 1-2\eps
\end{equation*}
\end{proof}
\begin{claim}\label{claim:monotonicitys-e}
	For any non-crucial edge $e\in \nc$ and $\beta < \alpha < 2\tau/p$, 
    $$\E\left[s_e\mid \tilde x_e = \alpha\right] \leq \E\left[ s_e \mid \tilde x_e = \beta \right]. $$
\end{claim}
\begin{proof}
    Consider for an arbitrary non-crucial edge $e=(u,v) \in \nc$. Let $Q^e$ be the set of samples containing $e$, i.e. $Q^{e}:=\{i \in [T] : e \in Q_i\}$ and $\gamma_i$ be a probability distribution on $\{0,1\}^2$ such that for a sample $(y^v, y^u) \sim \gamma_i$ we have
    $$y^v = \mathbbm{1}[(v,u') \in Q_i \text{ for some } u'\neq u] \qquad y^u = \mathbbm{1}[(v',u) \in Q_i \text{ for some } v'\neq v]$$
      when we condition on $e \notin Q_i$. Observe that $\gamma_i$'s are independent and identical distributions defined on a probability space when $e\notin Q_i$. Let $(y^v_i, y^u_i) \sim \gamma_i$ for all $i \in [T]$. For any $\theta_v > 0$ and $\theta_u > 0$
   
    \begin{align*}
        \Pr\left[\tilde x_v < \theta_v \bigwedge \tilde x_u < \theta_u \mid \tilde x_e = \beta\right] 
            &= \Pr\left[ \tilde x_e + \sum_{i \in [T]\setminus {Q^e}} \frac{y^v_i}{T} < \theta_v \bigwedge \tilde x_e + \sum_{i \in [T]\setminus {Q^e}}\frac{y^u_i}{T} < \theta_u \mid \tilde x_e=\beta\right]\\
            &= \Pr\left[ \beta + \sum_{i=1}^{T(1-\beta)} \frac{y^v_i}{T} < \theta_v \bigwedge \beta + \sum_{i=1}^{T(1-\beta)} \frac{y^u_i}{T} < \theta_u\right]\\
            &\geq \Pr\left[ \alpha + \sum_{i=1}^{T(1-\alpha)} \frac{y^v_i}{T} < \theta_v \bigwedge  \alpha + \sum_{i=1}^{T(1-\alpha)} \frac{y^u_i}{T} < \theta_u\right]\\
            &= \Pr\left[ \alpha + \sum_{i \in [T]\setminus {Q^e}} \frac{y^v_i}{T} < \theta_v \bigwedge \alpha + \sum_{i \in [T]\setminus {Q^e}} \frac{y^u_i}{T} < \theta_u \mid |{Q^e}| = T \cdot \alpha\right]\\
            &= \Pr\left[ \tilde x_e + \sum_{i \in [T]\setminus {Q^e}} \frac{y^v_i}{T} < \theta_v \bigwedge \tilde x_e + \sum_{i \in [T]\setminus {Q^e}} \frac{y^u_i}{T} < \theta_u \mid \tilde x_e=\alpha\right]\\
            &=\Pr\left[\tilde x_v < \theta_v \bigwedge \tilde x_u < \theta_u \mid \tilde x_e = \alpha\right].
    \end{align*}
The second equality and the second last equality follows from the fact that $y_i^v$ for all $i \in[T]\setminus Q^e$ are independent and identically distributed. The inequality follows becasue $\alpha - \beta  - \sum_{i=1}^{T(\alpha - \beta )}\frac{y_i^v}{T}\geq 0$ and $\alpha - \beta  - \sum_{i=1}^{T(\alpha - \beta )}\frac{y_i^u}{T}\geq 0$. Combining the above inequality, for any $\theta >0$, we obtain,
\begin{align*}
 \Pr[s_e \geq \theta \mid \tilde x_e = \beta] \geq  \Pr[s_e \geq \theta \mid \tilde x_e = \alpha] \quad \text{and} \quad
\E\left[s_e \mid \tilde x_e=\beta\right] \geq \E\left[s_e \mid \tilde x_e = \alpha\right].
\end{align*} 
    
\end{proof}

\begin{claim}\label{claim:monotonicitys}
	For any non-crucial edge $e=(u,v)\in \nc$,
    $$ \E\left[s_e \mid \tilde x_e = \frac{2\tau}{p}-\frac{1}{T}\right] \geq \E[s_e'] - \eps \cdot q_e . $$ 
\end{claim}
\begin{proof}
    \begin{align*}
    \E[s_e'] &= \E\left[s_e' \mid \tilde x_e = \frac{2\tau}{p} \right] \cdot \Pr\left[\tilde x_e = \frac{2\tau}{p}\right] + \E\left[s_e' \mid \tilde x_e < \frac{2\tau}{p} \right] \cdot \Pr\left[\tilde x_e < \frac{2\tau}{p}\right]\\
  &\leq  \Pr\left[\tilde x_e = \frac{2\tau}{p}\right] + \E\left[s_e' \mid \tilde x_e < \frac{2\tau}{p} \right]\\
      &\leq \eps \cdot q_e + \E\left[s_e' \mid \tilde x_e < \frac{2\tau}{p} \right],
    \end{align*}

    where the last inequality follows from Claim~\ref{claim:a2}. Hence, we obtain
    	\begin{align*}
        \E\left[s_e' \mid \tilde x_e < \frac{2\tau}{p} \right] \geq \E[s_e'] - \eps \cdot q_e.
        \end{align*}
    Set $\alpha=\frac{2\tau}{p}-\frac{1}{T}$, then for any $\theta_u, \theta_v >0$, we have
    \begin{align*}
        \Pr\left[\tilde x_v < \theta_v \bigwedge \tilde x_u < \theta_u \mid \tilde x_e = \alpha\right]
        &= \Pr\left[ \alpha + \sum_{i=1}^{T(1-\alpha)} \frac{y^v_i}{T} < \theta_v \bigwedge  \alpha + \sum_{i=1}^{T(1-\alpha)} \frac{y^u_i}{T} < \theta_u\right]\\
        &\geq \Pr\left[ \alpha + \sum_{i=1}^{T(1-\alpha)} \frac{y^v_i}{T} < \theta_v \bigwedge  \alpha + \sum_{i=1}^{T(1-\alpha)} \frac{y^u_i}{T} < \theta_u\right] \cdot \sum_{0\leq \beta\leq \alpha} \Pr[\tilde x_e = \beta ]\\
        &\geq \sum_{0\leq \beta \leq \alpha} \Pr\left[ \frac{2\tau}{p} + \sum_{i=1}^{T(1-\beta)} \frac{y^v_i}{T} < \theta_v \bigwedge \frac{2\tau}{p} + \sum_{i=1}^{T(1-\beta)} \frac{y^u_i}{T} < \theta_u \right] \cdot \Pr[\tilde x_e = \beta]\\
        &= \sum_{0\leq \beta \leq  \alpha} \Pr\left[ \frac{2\tau}{p} + \tilde x_{v \setminus e} < \theta_v \bigwedge \frac{2\tau}{p} + \tilde x_{u \setminus e }  < \theta_u \mid \tilde x_e = \beta\right] \cdot \Pr[\tilde x_e = \beta]\\
        &= \Pr\left[\frac{2\tau}{p} + \tilde x_{v \setminus e} < \theta_v \bigwedge \frac{2\tau}{p} + \tilde x_{u \setminus e} < \theta_u \mid \tilde x_e < \frac{2\tau}{p}\right].
      \end{align*}
    The first inequality holds becasue $\sum_{0\leq \beta\leq \alpha}\Pr[\tilde x_e = \beta ] \leq 1$. The second ineqialty holds because $\frac{2\tau}{p}>\alpha $. Moreover, the first equality and the second last equality follows from the fact that $y_i^v$ for all $i \in[T]\setminus Q^e$ are independent and identically distributed.  Therefore, by defination of $s_e$ and $s_e'$, we obatin, 
    \[
    	\Pr[s_e \geq \theta \mid \tilde x_e = \alpha] \geq  \Pr\left[s'_e \geq \theta \mid \tilde x_e < \frac{2\tau}{p} \right]
    \]
    and so
    \[
    	\E\left[s_e \mid \tilde x_e=\alpha\right] \geq \E\left[s'_e \mid \tilde x_e < \frac{2\tau}{p}\right]\geq \E[s_e'] - \eps \cdot q_e. 
    \]

\end{proof}
We are now ready to prove the main theorem. 

\begin{proof}[Proof of Lemma~\ref{lem:non-crucial}]
For the sake of exposition, we let $\alpha =\frac{2\tau}{p}-\frac{1}{T}  $. 
\begin{align*}
	\Pr[e\in M_\nc] &= \E_{\tilde x_e}[ \tilde x_e \E[s_e \mid \tilde x_e]] \cdot \Pr[e\in R]\\
	&\geq \sum_{\tilde x_e \leq \alpha}\tilde x_e \cdot \E[s_e \mid \tilde x_e ] \Pr[\tilde x_e] \cdot \Pr[e\in R]\\
	&\geq \sum_{\tilde x_e \leq \alpha}\tilde x_e \cdot \E\left[s_e \mid \tilde x_e =\alpha  \right] \Pr[\tilde x_e] \cdot \Pr[e\in R]&& (\text{Claim~\ref{claim:monotonicitys-e}})\\
	&\geq \E_{\tilde x_e} \left[\tilde x_e \mid \tilde x_e \leq \alpha \right] \cdot \E\left[s_e \mid \tilde x_e =\alpha  \right] \cdot \Pr[e\in R]\\
	&\geq \E\left[s_e \mid \tilde x_e =\alpha  \right]\cdot (1-3\eps) \cdot q_e && (\text{Claim~\ref{claim:a2}})\\
	&\geq (\E\left[s'_e\right] - \eps \cdot q_e) \cdot (1-3\eps) \cdot q_e &&  (\text{Claim~\ref{claim:monotonicitys}})\\
	&\geq (1-2\eps )^2\cdot (1-3\eps) \cdot q_e - \eps \cdot q_e &&(\text{Claim~\ref{claim:s-eis'1}})\\
	&\geq (1-12\eps)q_e.
\end{align*}




    
    Lastly, we will prove the monotonicity property \eqref{eq:monotone} of the matching constructed by Algorithm~\ref{alg:non-crucial-matching}.
Let $e\in \nc$ be an arbitrary non-crucial edge and $S \subseteq N(e)$ be an arbitrary subset of edges incident to $e$. Note that when we condition on $Q_1 \dots Q_T$ then $f_e := \sum_{i=1}^T \mathbbm{1}[e \in \optsol(Q_i)]/T$ becomes deterministic for all non-crucial edges $e \in \nc$. For the rest of the proof, we condition on the set $Q_1,\dots, Q_T$. 

Recall that $\Pr[e \in M_\nc] = \Ex[\tilde x_e \cdot s_e]$, therefore it is enough to show that $s_e$ weakly increases once we condition on $S \cap R=\emptyset$, i.e. all edges in the set $S$ are inactive. Note that the probability and expectation are taken over the random set $R$.  Formally, we let $\mathcal D$ be the distribution of active edges $R\subseteq E$ and $\D'$ be the distribution of active edges $R\subseteq E\setminus S$ conditioned on $S\cap R = \emptyset$. We will show that,
\begin{equation}
	\label{eq:stoc_dominance}
	\E_{R \sim \mathcal D}[s_e] \leq \E_{R' \sim \mathcal D'}[s_e].
\end{equation}
We couple two probability distributions $\mathcal D$ and $\mathcal D'$ as follows: we first sample $R \sim \mathcal D$ and construct $R' = R \setminus S$. Since each edge $e\in E$ appears in the active set of edges $R\sim \D$ independently with probability $p$, it follows that $R'\sim \D'$. Therefore the coupling $(\mathcal D, \mathcal D')$ is valid with the correct marginals. We denote this joint distribution of $(\D,\D')$ as $\gamma$. 

For any sample $(R, R')\sim \gamma$, from the definition of $\tilde x_{v} ( R) \geq \tilde x_v ( R')$, where $\tilde x_v (R)$ denotes the value of $\tilde x_v$ conditioned on the set of active edges being $R$. Note that $\tilde x_e ( R)$ and $\tilde x_v (R')$ are deterministic as we have already conditioned on $Q_1 ,\dots , Q_T$. This further implies that $s_e ( R )\leq s_e (R')$ where $\tilde x_v ( R')$, $s_e ( R )$ and $ s_e (R')$ are defined similar to the $\tilde x_v (R)$.  This completes the proof. 
    
    

\end{proof}

\subsection{Proof of Proposition~\ref{prop:high_prob_events}}
\begin{claim}\label{claim:bounded_deg}
Given an instance of stochastic weighted matching $\match$ with $p\leq \eps^4$ and $0<\eps<1/2$, for any non-crucial edge $e=(u,v)\in \nc$ with $N(e) = \{e_{1},\dots, e_{k}\}$ be the incident edges on vertices $u$ and $v$ in the graph $G$, $$\Pr\left[ Q_{\crs}\cap N(e) \leq \frac{(1+\eps)\cdot 2}{p}\right] \geq 1-\epsilon$$
\end{claim}

\begin{proof}
Let $e:=(u,v) \in \nc$ be an arbitrary non-crucial edge of the problem $\match$. For each edge $e_i\in N(e)$, let $Z_{i}$ be the indicator random variable which denotes whether the edge $e_i$ belongs to $Q_{\crs}$ or not, i.e., $Z_{i} = \mathbbm{1}[e_i \in Q_{\crs} \cap R]$. Note that $Z_i\sim Ber(\frac{q_e}{p})$ and independent across $i\in [k]$ and $\sum_{i\in [k]}\E[Z_i] = \sum_{i\in [k]} q_e /p \leq 2/p$. Therefore, by setting $\eps' = \frac{2(1+\eps) - \sum_{i\in [k]} q_{e_i}}{\sum_{i\in [k]} q_{e_i}}$ and applying Chernoff bound, we get
\begin{align*}
    \Pr\left[\sum_{i=1}^k Z_i \geq (1+\eps')\cdot \sum_{i=1}^k \frac{q_{e_i}}{p}\right] \leq \exp\left(-(\eps')^2 \cdot  \frac{\sum_{i\in [k]}q_{e_i}}{3p}\right).
\end{align*}
Now, as $\eps ' \geq \frac{2\cdot \eps}{\sum_{i\in [k]} q_{e_i}}$, we obtain,
\begin{align*}
        \Pr\left[\sum_{i=1}^k Z_i \geq (1+\eps')\cdot \sum_{i=1}^k \frac{q_{e_i}}{p}\right]&= \Pr\left[\sum_{i=1}^k Z_i \geq (1+\eps)\cdot  \frac{2}{p}\right] \\
        &\leq   \exp\left(-4\eps^2 \cdot  \frac{1}{3\eps^4\sum_{i\in [k]} q_{e_i}}\right)\\
        &\leq \exp\left( -\frac{2}{3\eps^2}\right) \leq \eps.
\end{align*}
\end{proof}

\begin{proof}[Proof of Proposition~\ref{prop:high_prob_events}]
We can express $\Pr[\mathcal E_e]$ as:
\begin{align*}
    \Pr[\mathcal E_e] &= \Pr[\mathcal E_e\mid e\notin  Q_{\crs}]\cdot \Pr[e\notin  Q_{\crs}] +  \Pr[\mathcal E_e\mid e\in  Q_{\crs}]\cdot \Pr[e\in  Q_{\crs}] \\
    &\leq \Pr[\mathcal E_e\mid e\notin  Q_{\crs}] + \Pr[e\in  Q_{\crs}] \leq \Pr[\mathcal E_e\mid e\notin  Q_{\crs}]  + \epsilon, \\
\end{align*}
where the last inequality follows due to definition of non-crucial edges, i.e. $\Pr[e\in Q_{\crs}] =q_e/p \leq \epsilon$. Combining above inequality with Claim~\ref{claim:bounded_deg}, we obtain
\begin{equation*}
    \Pr[\mathcal E_e\mid e\notin  Q_{\crs}] \geq \Pr[\mathcal E_e] -\epsilon \geq 1-2\epsilon.
\end{equation*}
\end{proof}

 \end{document}